\documentclass[11pt]{article}
\usepackage[letterpaper]{geometry}
\usepackage{amsmath,amsthm,amsfonts,amssymb,bbm}
\usepackage{enumerate,color}
\usepackage{graphicx}
\usepackage{subfigure}
\usepackage{epstopdf}
\usepackage{float}
\usepackage[affil-it]{authblk}
\usepackage{natbib}

\numberwithin{equation}{section}
\theoremstyle{plain}

\newtheorem{theorem}{Theorem}

\newtheorem{lemma}[theorem]{Lemma}

\newtheorem{remark}[theorem]{Remark}

\newtheorem{assumption}[theorem]{Assumption}

\title{Asymptotic Smiles for an Affine Jump-Diffusion Model}

\author[1]{Nian Yao}
\author[2,*]{Zhiqiu Li}
\author[1]{Zhichao Ling}
\author[1]{Junfeng Lin}

\affil[1]{College of Mathematics and Statistics, Shenzhen University, Shenzhen 518060, Guangdong, China}
\affil[2]{Department of Mathematics, Florida State University, Tallahassee, FL 32304, USA}
\affil[*]{Corresponding author zli@math.fsu.edu}
\affil[1]{yaonian@szu.edu.cn\\
1800202002@email.szu.edu.cn\\
2016020338@email.szu.edu.cn}

\begin{document}

\maketitle

\begin{abstract}
In this paper, we study the asymptotic behaviors of implied volatility of an affine jump-diffusion model. 
Let log stock price under risk-neutral measure follow an affine jump-diffusion model, we show that an explicit form of moment generating function for log stock price can be obtained by solving a set of ordinary differential equations. A large-time large deviation principle for log stock price is derived by applying the G\"{a}rtner-Ellis theorem. We characterize the asymptotic behaviors of the implied volatility in the large-maturity and large-strike regime using rate function in the large deviation principle. The asymptotics of the Black-Scholes implied volatility for fixed-maturity, large-strike and fixed-maturity, small-strike regimes are also studied. Numerical results are provided to validate the theoretical work.

\end{abstract}

\section{Introduction}\label{intro}
${}\quad$ Point process models the arrival times of events in many applications. Affine point process (or affine jump-diffusion model, or affine point process driven by a jump-diffusion) is a point process whose event arrival intensity is driven by an affine jump-diffusion (\cite{duffie2000transform}).
An affine point process can be further characterized as self-exciting or mutual-exciting. A self-exciting process means a jump increases the probabilities of occurrence of future jumps in the same component; while a mutual-exciting process increases the jump intensity in other components as well.

Because the affine point process has computational tractability, there have been many applications in finance and economics, such as \cite{errais2010affine, zhang2015affine, ait2015modeling, zhang2018affine, gao2019affine}.
\cite{errais2010affine} used affine point processes to model the cumulative losses due to corporate defaults in a portfolio. They assumed jump occurrence times are default times; while the jump sizes are the portfolio losses at defaults. They used index and tranche swap rates before and after Lehman Brothers' bankruptcy to conduct a market calibration study. Their results indicated the empirical importance of self-exciting property of a loss process. Meanwhile, they showed a simple affine point process is able to capture the implied default correlations during the month when Lehman defaulted.
\cite{ait2015modeling} observed jumps in stock markets extend over hours or days and across multiple markets. They concluded that a self-exciting (in time) and mutual-exciting (in space) process is capable of capturing such clustering patterns.
\cite{zhang2015affine} established a central limit theorem and a large deviation principle for affine point processes. By using these limits, they derived closed-form approximations to the distribution of an affine point process. The large deviation principle helped to construct an importance sampling scheme for estimating tail probabilities.
\cite{zhang2018affine} developed stochastic stability conditions for affine jump-diffusion processes.
By imposing a strong mean-reversion condition and a mild condition on the jump distribution, they established ergodicity for affine jump-diffusion processes.
They proved strong laws of large numbers and functional central limit theorems for additive functionals for this class of models.
When a closed-form solution of the characteristic function of an affine jump-diffusion process is not available, pricing evaluation requires a numerical solution of a set of ODEs within a numerical inversion routine. However, it is computationally expensive as the numerical transform inversion evokes thousands of calculations and each calculation requires the solution of a system of ODEs.
Later \cite{gao2019affine} extended the large-time limit theorems in \cite{zhang2015affine}. They derived large-time asymptotic expansions in large deviations and refined central limit theorem for affine point processes. They proposed a new approach based on the mod-$\phi$ convergence theory and they obtained the precise large deviations and refined central limits for an affine point process simultaneously. By truncating the asymptotic expansions, they obtained an explicit approximation for large deviation probabilities and tail expectations; meanwhile, such explicit approximation can be used as importance sampling in Monte Carlo simulations.

Affine point process includes the linear Markovian Hawkes process as a special case \cite{hawkes1971spectra, hawkes1971point}. Hawkes process has wide range of applications in various domains such as seismology \cite{ogata1988statistical}, genome analysis \cite{reynaud2010adaptive}, social network \cite{crane2008robust}, modeling of crimes \cite{mohler2011self} and finance \cite{bacry2015hawkes} (\cite{bacry2015hawkes} provided a comprehensive survey of applications of Hawkes process in finance).

Option pricing problems have been well studied when the underlying follows a jump-diffusion process. Back to the 1970s, \cite{merton1976option} proposed a jump-diffusion process and assumes the jump size follows a log-normal distribution. They showed a European option can be written as a weighted sum of Black-Scholes European option prices. Later \cite{kou2002jump} assumed the jump size follows a double exponential distribution and a closed-form solution was provided.

As to the underlying follows an affine jump-diffusion point process or has Hawkes jumps, option pricing problems are much less studied. This is because of the closed-form solution of option pricing is no longer available. For instance, \cite{ma2017pricing} studied a vulnerable European option pricing problem assuming underlying asset and option writer's asset value both following the Hawkes processes. However, as the analytic solutions are unavailable, they implemented the thinning algorithm to compare the proposed model performance versus other models. 


There have been studies on option pricing problems at asymptotic regimes.
\cite{FordeJacquier2011} studied the large-time asymptotic behaviors of European call and put options under the Heston stochastic volatility model. They derived the large-time large deviation principle for the log return of underlying over time-to-maturity by applying the G\"{a}rtner-Ellis theorem. At the same time, they derived the asymptotic Black-Scholes implied volatility at large-time. Later \cite{jacquier2016large} characterizes the forward implied volatility smiles for the same model. Similar work has been extended to other stochastic volatility models, such as the SABR and CEV-Heston models (\cite{forde2013large}), a class of affine stochastic volatility models (\cite{jacquier2013large}) and multivariate Wishart stochastic volatility models (\cite{alfonsi2019long}).

\cite{Lee2004} studied the asymptotics of the Black-Scholes implied volatility in the regime where maturity $T$ is fixed and strike is large or small. They showed the large-strike tail of the implied volatility skew is bounded by $O(|x|^{1/2})$, where $x$ is log-moneyness.
They proved the explicit moment formula that determines the smallest coefficient in that bound for a given $T$.
In addition, they pointed out the linkage between finite moments and tail slopes is model-independent. Some applications of moment formula such as skew extrapolation and model calibration were discussed.

In this paper, we study the asymptotic behaviors of the implied volatility of an affine jump-diffusion model. This article is organized as follows: In Section 2.1, we express the moment generating function of the affine jump-diffusion model as solutions of a set of ordinary differential equations by using the Feynman-Kac formula. In Section 2.2, we obtain the large-time large deviation principle of the log return of the stock price under the risk-neutral measure by using G\"{a}rtner-Ellis theorem. In Section 2.3, we characterize the asymptotic behaviors of the implied volatility in the large-maturity and large-strike regime using rate function in the large deviation principle. In Section 2.4, we study the asymptotic of the implied volatility for fixed-maturity, large-strike and fixed-maturity small-strike regimes. In Section 3, we conduct numerical studies to validate the theoretical work. Lastly, conclusion remarks are in Section 4.

\section{Affine jump-diffusion model}
${}\quad$ We assume the underlying stock $S_t$ under the risk-neutral measure $\mathbb{Q}$ follows an affine jump-diffusion model:
\begin{equation}\label{affine point dynamics}
\frac{dS_{t}}{S_{t-}}
=\sigma dW_{t}^{\mathbb{Q}}
+(dJ_{t}-\lambda_{t}^{N}\mu_{Y}dt),
\end{equation}
where
\begin{equation}
J_{t}=\sum_{i=1}^{N_{t}}(e^{Y_{i}}-1),
\end{equation}
where $Y_{i}$ are i.i.d. random jump sizes independent of $N_{t}$ and $W_t^\mathbb{Q}$ and $\mu_{Y}=\mathbb{E}[e^{Y}]-1$. $Y_{i}$ follows a probability distribution $Q(da)$. We assume that $N_{t}$ is an affine point process which has intensity
$\lambda^{N}_{t}=\alpha+\beta \lambda_t$ at $t > 0$ and $\lambda_t$ satisfies the dynamics:
\begin{equation}\label{APP}
d\lambda_{t}=b(c-\lambda_{t})dt+\sigma\sqrt{\lambda_{t}}dB_{t}
+a dN_{t}.
\end{equation}

We make following basic assumptions that are required for modelling an affine jump-diffusion model (\cite{ZLJ2014B}):
\begin{assumption}\label{1}
\begin{itemize}
 \item[1.] $a,b,c,\alpha,\beta,\sigma > 0$.
 \item[2.] $b>a\beta$. This condition indicates that there exists a unique
stationary process $\lambda^\infty$ which satisfies the dynamics (\ref{APP}).
 \item[3.] $2bc\ge\sigma^2$. This condition implies that $\lambda_t\ge 0$ with probability 1.
\end{itemize}
\end{assumption}

Also we assume that $B_{t}$ is independent of $W_{t}^{\mathbb{Q}}$.
One should notice that, the point process $N_t$ reduces to a linear Hawkes process with an exponential decay kernel when the Brownian motion term $B_t=0$.
If $adN_t=0$, then the process $\lambda_t$ reduces to a Cox–Ingersoll–Ross process.
The log stock price under the risk-neutral measure via $S_{t}=S_{0}e^{X_{t}}$ is
\begin{equation}
X_{t}=-\frac{1}{2}\sigma^{2}t+\sigma W_{t}^{\mathbb{Q}}-\mu_{Y}\int_{0}^{t}\lambda^N_{s}ds+\sum_{i=1}^{N_{t}}Y_{i}.
\end{equation}
We can write $N_t=\sum_{i=1}\mathbbm{1}_{\{T_i\le t\}}$ and $L_t=\sum_{i\ge 1}Y_i\mathbbm{1}_{\{T_i\le t\}}$ where $T_n$ is the n-th jump time of $N_t$.
The two-dimensional process $(\lambda,L)$ is Markovian on $D={\mathbb{R}_+}\times{\mathbb{R}}$ with an infinite generator given by
\begin{equation}\label{generator}
\mathcal{L}f(\lambda, L)=b(c-\lambda)\frac{\partial f}{\partial \lambda}+\frac{1}{2}\sigma^2\lambda\frac{\partial^2 f}{\partial \lambda^2}+ (\alpha+\beta\lambda)\int_{\mathbb{R}}(f(\lambda+a,L+y)-f(\lambda,L))Q(dy)
\end{equation}
for a given function $f:{\mathbb{R}_+}\times{\mathbb{R}}\rightarrow\mathbb{R}$ with twice continuously differentiable and for all $\lambda\in \mathbb{R}_+$, $|\int_{\mathbb{R}}f(L+y,\lambda+a)Q(dy)|<\infty$.

\subsection{Moment generating function for $X_t$}
${}\quad$ In this section, we compute the moment generating function for $X_t$. The result is summarized in following Lemma \ref{mgf_X}.
\begin{lemma}\label{mgf_X}
The moment generating function for $X_t$ is
\begin{equation}
\mathbb{E}[e^{\theta X_{t}}]
= e^{(-\frac{1}{2}\theta\sigma^{2}+\frac{1}{2}\theta^{2}\sigma^{2}-\theta\mu_Y\alpha)t
+D(t;\Theta)\lambda+\theta_3L+F(t;\Theta)}
\end{equation}
where $\theta\in\mathbb{R}$, $\Theta=(\theta_1,\theta_2,\theta_3)\in \mathbb{R}^3$ and $D(t;\Theta)$, $F(t;\Theta)$ satisfy the following ordinary differential equations
\begin{equation}
\left\{\begin{array}{r@{}l@{\qquad}l}
& D'(t;\Theta)+bD(t;\Theta)-\frac{1}{2}\sigma^2D^2(t;\Theta)-\beta\int_{\mathbb{R}}(e^{D(t;\Theta)a+\theta_3y}-1)Q(dy)-\theta_1=0,\\
& F'(t;\Theta)-bcD(t;\Theta)-\alpha\int_{\mathbb{R}}(e^{D(t;\Theta)a+\theta_3y}-1)Q(dy)=0,\\
& D(0;\Theta)=\theta_2, F(0;\Theta)=0.
\end{array}\right.
\end{equation}
\end{lemma}

\begin{proof}
Given any $\theta$ in $\mathbb{R}$, the moment generating function for $X_t$ is
\begin{equation}\label{moment gen_X affine}
\begin{aligned}
\mathbb{E}[e^{\theta X_{t}}]
&=\mathbb{E}\left[e^{\theta\left(-\frac{1}{2}\sigma^{2}t+\sigma W_{t}^{\mathbb{Q}}-\mu_{Y}\int_{0}^{t}\lambda^N_{s}ds+\sum_{i=1}^{N_{t}}Y_{i}\right)}\right]\\
&=e^{(-\frac{1}{2}\theta\sigma^{2}+\frac{1}{2}\theta^{2}\sigma^{2}-\theta\mu_Y\alpha)t}
\mathbb{E}[e^{-\theta\mu_{Y}\beta\int_{0}^{t}\lambda_{s}ds+\theta L_t}].\\
\end{aligned}
\end{equation}
For any $\Theta=(\theta_1,\theta_2,\theta_3)\in \mathbb{R}^3$, we assume
\begin{equation}\label{Fey-Kac}
\mathbb{E}[e^{\theta_1 \int_{t}^{T}\lambda_{s}ds+\theta_2\lambda_{T}+\theta_3L_T}|\lambda_{t}=\lambda,L_t=L]=u(t,\lambda,L):=u(t,\lambda,L,\Theta).
\end{equation}
By applying Feynman-Kac formula, we have
\begin{equation}\label{Fey-Kac int and L}
\left\{\begin{array}{r@{}l@{\qquad}l}
& \frac{\partial u}{\partial t}+b(c-\lambda)\frac{\partial  u}{\partial \lambda}\\
&+\frac{1}{2}\sigma^2\lambda\frac{\partial^2  u}{\partial \lambda^2}+ (\alpha+\beta\lambda)\int_{\mathbb{R}}( u(t,\lambda+a,L+y)- u(t,\lambda,L))Q(dy)+\theta_1\lambda u=0,\\
& u(T,\lambda,L,\Theta)=e^{\theta_2\lambda+\theta_3L}.
\end{array}\right.
\end{equation}
Let us try a solution in the form of 
$u(t,\lambda,L)=e^{A(t;\Theta)\lambda+B(t;\Theta)L+C(t;\Theta)}$, then $A(t;\Theta), B(t;\Theta), C(t;\Theta) $ satisfy the following ordinary differential equations
\begin{equation}
\left\{\begin{array}{r@{}l@{\qquad}l}
& A'(t;\Theta)-bA(t;\Theta)+\frac{1}{2}\sigma^2A^2(t;\Theta)+\beta\int_{\mathbb{R}}(e^{A(t;\Theta)a+B(t;\Theta)y}-1)Q(dy)+\theta_1=0,\\
& B'(t;\Theta)=0,\\
& C'+bcA(t;\Theta)+\alpha\int_{\mathbb{R}}(e^{A(t;\Theta)a+B(t;\Theta)y}-1)Q(dy)=0,\\
& A(T;\Theta)=\theta_2, B(T;\Theta)=\theta_3, C(T;\Theta)=0.
\end{array}\right.
\end{equation}
Then we have $u(s,\lambda,L)=e^{A(s;\Theta)\lambda+\theta_3L+C(s;\Theta)}$ and $A(s;\Theta), C(s;\Theta)$ satisfy the following ordinary differential equations
\begin{equation}\label{u}
\left\{\begin{array}{r@{}l@{\qquad}l}
& A'(t;\Theta)-bA(t;\Theta)+\frac{1}{2}\sigma^2A^2(t;\Theta)+\beta\int_{\mathbb{R}}(e^{A(t;\Theta)a+\theta_3y}-1)Q(dy)+\theta_1=0,\\
& C'+bcA(t;\Theta)+\alpha\int_{\mathbb{R}}(e^{A(t;\Theta)a+\theta_3y}-1)Q(dy)=0,\\
& A(T;\Theta)=\theta_2, C(T;\Theta)=0.
\end{array}\right.
\end{equation}
Let $f(t,\lambda,L):= f(t,\lambda,L,\Theta):= \mathbb{E}[e^{\theta_1 \int_{0}^{t}\lambda_{s}ds+\theta_2\lambda_{t}+\theta_3L_t}|\lambda_0=\lambda,L_0=L]$.
Let $u(t,\lambda,L)=f(T-t,\lambda,L)$ and make the time change $t \mapsto T-t$
to change the backward equation to the forward equation, we have
\begin{equation}\label{Fey-Kac int and L1}
\left\{\begin{array}{r@{}l@{\qquad}l}
&-\frac{\partial f}{\partial s}+b(c-\lambda)\frac{\partial  f}{\partial \lambda}\\
&+\frac{1}{2}\sigma^2\lambda\frac{\partial^2  f}{\partial \lambda^2}+ (\alpha+\beta\lambda)\int_{\mathbb{R}}(f(s,\lambda+a,L+y)- f(s,\lambda,L))Q(dy)+\theta_1\lambda f=0,\\
& f(0,\lambda,L,\Theta)=e^{\theta_2\lambda+\theta_3L}.
\end{array}\right.
\end{equation}
We try $f(s,\lambda,L)=e^{D(s;\Theta)\lambda+E(s;\Theta)L+F(s;\Theta)}$, then we have $D(s;\Theta),E(s;\Theta), F(s;\Theta)$ satisfy the following ordinary differential equations
\begin{equation}\label{Fey-Kac int and L2}
\left\{\begin{array}{r@{}l@{\qquad}l}
& D'(t;\Theta)+bD(t;\Theta)-\frac{1}{2}\sigma^2D^2(t;\Theta)-\beta\int_{\mathbb{R}}(e^{D(t;\Theta)a+E(t;\Theta)y}-1)Q(dy)-\theta_1=0,\\
& E'(t;\Theta)=0,\\
& F'-bcD(t;\Theta)-\alpha\int_{\mathbb{R}}(e^{D(t;\Theta)a+E(t;\Theta)y}-1)Q(dy)=0,\\
& D(0;\Theta)=\theta_2, E(0;\Theta)=\theta_3, F(0;\Theta)=0.
\end{array}\right.
\end{equation}
Finally we have $f(s,\lambda,L)=e^{D(s;\Theta)\lambda+\theta_3L+F(s;\Theta)}$ and $D(s;\Theta)$, $F(s;\Theta)$ satisfy the following ordinary differential equations
\begin{equation}\label{eqn:mgf_ODE}
\left\{\begin{array}{r@{}l@{\qquad}l}
& D'(s;\Theta)+bD(s;\Theta)-\frac{1}{2}\sigma^2D^2(s;\Theta)-\beta\int_{\mathbb{R}}(e^{D(s;\Theta)a+\theta_3y}-1)Q(dy)-\theta_1=0,\\
& F'(s;\Theta)-bcD(s;\Theta)-\alpha\int_{\mathbb{R}}(e^{D(s;\Theta)a+\theta_3y}-1)Q(dy)=0,\\
& D(0;\Theta)=\theta_2, F(0;\Theta)=0.
\end{array}\right.
\end{equation}
\end{proof}

\subsection{Large deviation principle for $X_t$}
${}\quad$ In this section, we derive the following theorem which describes the large-time large deviation asymptotic behaviors of the log stock price.
This result will be used later to derive the asymptotics for option pricing and implied volatility smiles in the regime where the maturity is large and the log-moneyness is of the same order as the maturity. We refer readers to \cite{DZ1998} for formal definition of large deviation principle and the applications.
\begin{theorem}\label{LDP for X affine} (Large Deviation Principle for $X_t$). Under Assumption 1, 
$\mathbb{Q}(\frac{1}{t}X_t\in\cdot)$
satisfies a large deviation principle on $\mathbb{R}$ with the rate function:
\begin{equation}\label{rate for affine X}
I(x)=\mathop{\sup}\limits_{\theta\in\mathbb{R}}\left\{\theta x-\Lambda(\theta)\right\},
\end{equation}
where $$\Lambda(\theta) = \left(\frac{1}{2}\sigma^2\theta^2-\left(\frac{1}{2}\sigma^2+\mu_Y\alpha\right)\theta+bcy(\theta)+\alpha\left(e^{ay(\theta)}
\mathbb{E}[ e^{\theta Y}]-1\right)\right)$$
and $y(\theta)$ is the smaller solution of the equation
\begin{equation}
-by+\frac{1}{2}\sigma^2y^2+\beta(\mathbb{E}[e^{ay+\theta Y}]-1)-\theta\mu_Y\beta=0,
\end{equation}
if solution exists. Otherwise $y(\theta)=+\infty$.
\end{theorem}

\begin{proof}
From (\ref{moment gen_X affine}) and (\ref{eqn:mgf_ODE}) we know $(\theta_1,\theta_2,\theta_3)=(-\theta\mu_Y\beta,0,\theta)$ and, for any $\theta\in \mathbb{R}$, we have:
\begin{equation}\label{moment gen_X}
\begin{aligned}
\mathbb{E}[e^{\theta X_{t}}]
&=e^{(-\frac{1}{2}\theta\sigma^{2}+\frac{1}{2}\theta^{2}\sigma^{2}-\theta\mu_Y\alpha)t}
\mathbb{E}[e^{-\theta\mu_{Y}\beta\int_{0}^{t}\lambda_{s}ds+\theta L_t}]\\
&=e^{(-\frac{1}{2}\theta\sigma^{2}+\frac{1}{2}\theta^{2}\sigma^{2}-\theta\mu_Y\alpha)t+\bar{D}(t,\theta)\lambda+\theta L+\bar{F}(t,\theta)}.
\end{aligned}
\end{equation}
where $\bar{D}(t;\theta)$ and $\bar{F}(t;\theta)$ satisfy the following ordinary differential equations
\begin{equation}\label{DF}
\left\{\begin{array}{r@{}l@{\qquad}l}
&\bar{D}'(t;\theta)+b\bar{D}(t;\theta)-\frac{1}{2}\sigma^2\bar{D}^2(t;\theta)-\beta\int_{\mathbb{R}}(e^{\bar{D}(t;\theta)a+\theta y}-1)Q(dy)+\theta\mu_Y\beta=0,\\
& \bar{F}'(t;\theta)-bc\bar{D}(t;\theta)-\alpha\int_{\mathbb{R}}(e^{\bar{D}(t;\theta)a+\theta y}-1)Q(dy)=0,\\
& \bar{D}(0;\theta)=0, \bar{F}(0;\theta)=0.
\end{array}\right.
\end{equation}

Thus, from (\ref{moment gen_X}) we have
\begin{align*}
\Lambda(\theta):&=\lim_{t\rightarrow\infty}\frac{1}{t}\log\mathbb{E}[e^{\theta X_{t}}]\\
                &=\frac{1}{2}\sigma^2\theta^2-\left(\frac{1}{2}\sigma^2+\mu_Y\alpha\right)\theta
                +\lambda\lim_{t\rightarrow\infty}\frac{\bar{D}(t;\theta)}{t}+\lim_{t\rightarrow\infty}\frac{\bar{F}(t;\theta)}{t},
\end{align*}

From (\ref{DF}), one can see that
\begin{equation*}
\begin{array}{r@{}l@{\qquad}l}
\Gamma(D,\theta):&=-bD+\frac{1}{2}\sigma^2 D^2+\beta\int_{\mathbb{R}}(e^{aD+\theta y}-1)Q(dy)-\theta\mu_Y\beta\\
    &=-bD+\frac{1}{2}\sigma^2D^2+\beta(\mathbb{E}[e^{aD+\theta Y}]-1)-\theta\mu_Y\beta.
\end{array}
\end{equation*}

Next we want to find the range of $\theta$ such that
\begin{equation}\label{Gamma function}
\Gamma(y,\theta)=-by+\frac{1}{2}\sigma^2y^2+\beta(\mathbb{E}[e^{ay+\theta Y}]-1)-\theta\mu_Y\beta=0
\end{equation}
has a solution of $y(\theta)$.
We know that
\begin{align*}
\Gamma_y'(y,\theta)&= -b +\sigma^2 y+ a\beta e^{ay}\mathbb{E}[e^{\theta Y}],\\
\Gamma_y''(y,\theta)&=\sigma^2+a^2\beta e^{ay} \mathbb{E}[e^{\theta Y}]
\end{align*}
and we find that $\Gamma_y''(y,\theta)>0$, so $\Gamma(y,\theta)$ is convex and $\Gamma_y'(y,\theta)$ is increasing in $y$. Clearly we have ${\lim\limits_{y \to -\infty}}\Gamma_y'(y,\theta)=-\infty$
and ${\lim\limits_{y \to +\infty}}\Gamma_y'(y,\theta)=+\infty$, so there exists a unique $y_c(\theta)$ which satisfies the following equation,
\begin{equation}\label{yc equation}
-b +\sigma^2 y_c+a\beta e^{ay_c}\mathbb{E}[e^{\theta Y}]=0.
\end{equation}
We take the derivative of $y_c(\theta)$ on $\theta$,
\begin{equation}
y_c'(\theta)=-\frac{a\beta e^{ay_c(\theta)}\mathbb{E}[Y e^{\theta Y}]}{\sigma^2+a^2\beta e^{a y_c(\theta)}\mathbb{E}[e^{\theta Y}]}\\
\end{equation}
And we can rewrite $\Gamma(y_c(\theta),\theta)$
\begin{equation}
\Gamma(y_c(\theta),\theta)=G(\theta):=-b y_c(\theta)+\frac{\sigma^2}{2} y_c^2(\theta) + \beta e^{a y_c(\theta)} \mathbb{E}[e^{\theta Y}]-\beta(\theta \mu_Y +1)
\end{equation}
Now we arrive at find the scope of $\theta$ such that $G(\theta)\le 0$. Take the derivative of $G(\theta)$ on $\theta$,
\begin{equation}
G'(\theta)=\beta\left( e^{a y_c(\theta)}\mathbb{E}[Ye^{\theta Y}]-\mu_Y\right)
\end{equation}
\begin{equation}
G''(\theta)=\frac{\sigma^2 \beta e^{ay_c(\theta)}\mathbb{E}[Y^2e^{\theta Y}]+a^2 \beta^2 e^{2ay_c(\theta)}(\mathbb{E}[Y^2e^{\theta Y}]\mathbb{E}[e^{\theta Y}]-\mathbb{E}[Ye^{\theta Y}]^2)}{\sigma^2 + a^2 \beta e^{ay_c(\theta)}\mathbb{E}[e^{\theta Y}]}
\end{equation}
By Cauchy-Schwarz inequality we can get $G''(\theta)>0$, so $G(\theta)$ is convex, and $G'(\theta)$ is increasing. Further, with the fact that ${\lim\limits_{\theta \to -\infty}}y_c(\theta)=\frac{b}{\sigma^2}$ from (\ref{yc equation}), we can easily see that ${\lim\limits_{\theta \to -\infty}}G'(\theta)<0$,
so we just need to judge whether $\theta_c$ exist such that $G'(\theta_c)=0$. We discusses in two cases.

{\bf Case one}: ${\lim\limits_{\theta \to +\infty}}G'(\theta)\le 0$, in this case, only ${\lim\limits_{\theta \to +\infty}}G(\theta)<0$ can make the function has a solution,
and the unique solution $\theta_{\min}$ satisfies
\begin{equation}\label{yc}
\left\{\begin{array}{r@{}l@{\qquad}l}
&\theta=\frac{2(a \beta +\sigma^2)y_c+\alpha \sigma^2 y_c^2-2a \beta+2b}{a \beta \mu_Y},\\
&-b+\sigma^2y_c+a\beta e^{a y_c}E[e^{\theta Y}]=0.
\end{array}\right.
\end{equation}

\begin{figure}[H]
  \centering
  \includegraphics[scale=0.5]{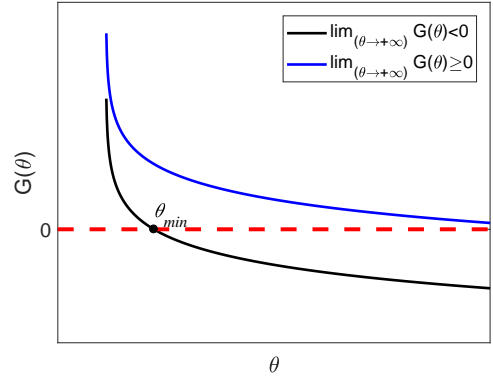}\\
  \caption{Case one}
\end{figure}

{\bf Case two}: ${\lim\limits_{\theta \to +\infty}}G'(\theta)>0$, in this case $G'(\theta_c)=0$ has a unique solution $\theta_c$. And $G(\theta_c)$ is the minimum of $G(\theta)$. We write $\theta_{\min}$ and $\theta_{\max}$
for the two solutions for equation
\begin{equation}\label{theta-min-max}
\left\{\begin{array}{r@{}l@{\qquad}l}
&G(\theta)=-b y_c(\theta)+\frac{\sigma^2}{2} y_c^2(\theta) + \beta e^{a y_c(\theta)} \mathbb{E}[e^{\theta Y}]-\beta(\theta \mu_Y +1)=0,\\
&-b+\sigma^2y_c+\alpha\beta e^{\alpha y_c}E(e^{\theta Y})=0.
\end{array}\right.
\end{equation}

\begin{figure}[H]
  \centering
  \includegraphics[scale=0.5]{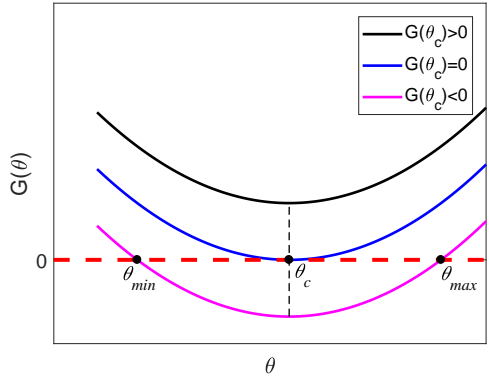}\\
  \caption{Case two}
\end{figure}

\begin{itemize}
 \item[1.] If ${\lim\limits_{\theta \to +\infty}}G(\theta)<0$, then when $\theta\ge\theta_{\min}$ in (\ref{yc}), $G(\theta)\le 0$.
 \item[2.] If ${\lim\limits_{\theta \to +\infty}}G'(\theta)>0$, then when $\theta\in [\theta_{\min}, \theta_{\max}]$, $G(\theta)\le 0$.
\end{itemize}

Therefore for $\theta\in [\theta_{\rm min},\theta_{\rm max}]$ (in {\bf Case one}, $\theta_{\rm max}\longrightarrow +\infty$), we have
$$\Lambda(\theta)=\mathop{\lim}\limits_{t\rightarrow\infty}
\frac{1}{t}\log\mathbb{E}[e^{\theta X_t}]=\frac{1}{2}\sigma^2\theta^2-\left(\frac{1}{2}\sigma^2+\mu_Y\alpha\right)\theta+bcy(\theta)+\alpha\left(e^{ay(\theta)}
\mathbb{E}[e^{\theta Y}]-1\right).$$
When $\theta\notin [\theta_{\rm min},\theta_{\rm max}]$, this limit is $\infty$.

We are to check two conditions for G${\rm \ddot{a}}$rtner-Ellis theorem. The first condition is essential smoothness. By differentiating the equation (\ref{Gamma function})
with respect to $\theta$, that is when $\theta\rightarrow \theta_{\rm min(max)}$, then $y\rightarrow y_c$, and
$$\frac{\partial y}{\partial\theta}=\frac{\beta(\mu_Y-e^{ay}\mathbb{E}[Ye^{\theta Y}])}{-b+\sigma^2y+a\beta e^{ay}\mathbb{E}[e^{\theta Y}]}\rightarrow+\infty.$$
The second is $0\in [\theta_{\rm min},\theta_{\rm max}]$. As $[\theta_{\rm min},\theta_{\rm max}]$ is the range of $\theta$ such that equation (\ref{Gamma function}) has a solution of $y(\theta)$.
When $\theta=0$, the equation becomes
\begin{equation}\label{Gamma function 0}
\Gamma(y,0)=-by+\frac{1}{2}\sigma^2y^2+\beta e^{ay}-\beta=0.
\end{equation}
It is straightforward to see that $y=0$ is the solution, therefore $0\in [\theta_{\rm min},\theta_{\rm max}]$.

Upon applying
G${\rm \ddot{a}}$rtner-Ellis theorem (refer to \cite{DZ1998} for the definition of essential smoothness and statement of
G${\rm \ddot{a}}$rtner-Ellis theorem), $\mathbb{Q}(\frac{1}{t}X_t\in\cdot)$ satisfies
a large deviation principle with rate function
\begin{equation*}
I(x)=\mathop{\sup}\limits_{\theta\in\mathbb{R}}\left\{\theta x-\left(\frac{1}{2}\sigma^2\theta^2-\left(\frac{1}{2}\sigma^2+\mu_Y\alpha\right)\theta+bc y(\theta)+\alpha\left(e^{ay(\theta)}
\mathbb{E}[e^{\theta Y}]-1\right)\right)\right\}.
\end{equation*}
\end{proof}

\subsection{Asymptotics of implied volatility in large-maturity and large-strike regime}
${}\quad$ In this section, we use the rate function in the large deviation principle for $X_t$ to characterize the asymptotic behaviours of implied volatility in large-maturity and large-strike regime.

Consider an European call option with maturity $T$ and strike $K$ is given as
\begin{equation}
C(K,T):=D(T)\mathbb{E}\left[\left(S_T-K\right)^+\right],\nonumber
\end{equation}
where $S_T$ is the underlying stock price at maturity $T$ and $D(T)$ is the discount factor. One should notice the corresponding put option price $P(K,T)$ can be found straightforwardly using call-put parity.
$C(K,T)$ indicates the dependence on the maturity $T$ and strike $K$.
Let $F_0=\mathbb{E}S_T$ be the forward price of underlying stock.
For a given $F_0$, the log moneyness $k$ is related to strike by
\begin{equation}\label{moneyness}
k:= \log(K/F_0),
\end{equation}
so $K(k)=F_0e^k$ is the strike at log moneyness $k$.
The Black-Scholes implied volatility with log moneyness $k$ and at maturity $T$ is defined as $\sigma_{BS}(k,T)$ which uniquely solves
\begin{equation}\label{call implied volatility}
C(K(k),T)=C^{BS}(k,\sigma_{BS}(k,T)),
\end{equation}
where
\begin{align*}
C^{BS}(k,\sigma)=D(T)\left(F_0\Phi(d_+)-K(k)\Phi(d_-)\right) \quad\text{and}\quad d_{\pm}=\frac{-k}{\sigma\sqrt{T}}\pm\frac{\sigma\sqrt{T}}{2},
\end{align*}
and $\Phi$ is the cumulative distribution function of a standard normal distribution. Similarly, for a European put option, its implied volatility $\sigma_{BS}(k,T)$ uniquely solves
\begin{equation}\label{put implied volatility}
P(K(k),T) = P^{BS}(k,\sigma_{BS}(k,T)),
\end{equation}
where
\begin{align*}
P^{BS}(k,\sigma)=D(T)(K(k)\Phi(-d_-)-F_0\Phi(-d_+)).
\end{align*}

\begin{theorem}\label{lmls} In the joint regime of large-maturity, large-strike with $k=\log(K/S_{0})$
$(T\rightarrow\infty$, $|k|\rightarrow\infty)$, the implied volatility $\sigma_{\text{BS}}(k,T)$
approaches the limit
\begin{equation}\label{eqn:joint regime of large T and large K}
\lim_{T\rightarrow\infty}\sigma^{2}_{\text{BS}}(xT,T)=\sigma_{\infty}^{2}(x),
\end{equation}
where
\begin{equation} \label{eqn:approx-implied-vol}
\sigma_{\infty}^{2}(x)=
\begin{cases}
2(2I(x)-x-2\sqrt{I^{2}(x)-xI(x)}) & x\in(-\infty,x_{L})\cup(x_{R},\infty)
\\
2(2I(x)-x+2\sqrt{I^{2}(x)-xI(x)}) & x\in[x_{L},x_{R}]
\end{cases}
\end{equation}
where $I(x)$ is defined in (\ref{rate for affine X}) and
\begin{equation}
x_{L}=-\left(\frac{1}{2}\sigma^2+\mu_{Y}\alpha\right)+\left(bc+a\alpha\right)\frac{\beta\left(\mu_{Y}-
\mathbb{E}[Y]\right)}{a\beta-b} + \alpha \mathbb{E}[Y],
\end{equation}
and
\begin{equation}
x_{R}=\left(\frac{1}{2}\sigma^2-\mu_{Y}\mathbb{E}[e^{Y}]\alpha\right)+\left(bc+a\mathbb{E}[e^Y]\alpha
\right)\frac{\mathbb{E}[e^{Y}]\beta\left(\mu_{Y}-\mathbb{E}[\bar{Y}]\right)}{a\mathbb{E}[e^{Y}]\beta-b} + \mathbb{E}[e^Y] \alpha \mathbb{E}[\bar{Y}],
\end{equation}
where $\bar{Y}$ follows the probability distribution $\frac{e^{Y}}{\mathbb{E}[e^{Y}]}d\mathbb{Q}$.
\end{theorem}

\begin{proof}
First, let us give a more explicit expression for $I(x)$ in (\ref{rate for affine X}).
Note that
\begin{equation*}
I(x)=\theta^{\ast}x-\Lambda(\theta^{\ast}),
\end{equation*}
Let $\frac{d}{d\theta}I(x)=0$, where $x=\Lambda'(\theta^{\ast})$ so that
\begin{equation*}
\sigma^{2}\theta^{\ast}-\left(\frac{1}{2}\sigma^{2}+\mu_{Y}\alpha\right)+bcD'(\theta^{\ast})+\alpha D'(\theta^{\ast})e^{aD}\mathbb{E}[e^{\theta^{\ast}Y}]+\alpha\mathbb{E}[Ye^{aD+\theta^{\ast}Y}]=x,
\end{equation*}
which gives that
\begin{equation*}
D'(\theta^{\ast})=\frac{x+\frac{1}{2}\sigma^{2}+\mu_{Y}\alpha-\theta^{\ast}\sigma^{2}-\alpha\mathbb{E}[Ye^{aD+\theta^{\ast}Y}]}{bc+\alpha e^{aD}\mathbb{E}[e^{\theta^{\ast}Y}]}.
\end{equation*}
On the other hand, take the derivative of equation $\Gamma(D(\theta),\theta)=0$ on $\theta$,
\begin{equation*}
-bD'(\theta)+\sigma^{2}D(\theta)D'(\theta)+\beta\mathbb{E}\left[(aD'(\theta)+Y)e^{aD(\theta)+\theta Y}\right]-\mu_{Y}\beta=0,
\end{equation*}
that is
\begin{equation*}
D'(\theta)\left(\sigma^{2}D(\theta)-b+a\beta\mathbb{E}[e^{aD(\theta)+\theta Y}]\right)=\mu_{Y}\beta-\beta\mathbb{E}[Ye^{aD(\theta)+\theta Y}].
\end{equation*}
Therefore we can solve for $\theta^{\ast}$ and $D(\theta^{\ast})$ from the following equations:
\begin{equation}\label{Dtheta}
\left\{\begin{array}{r@{}l@{\qquad}l}
&\frac{x+\frac{1}{2}\sigma^{2}+\mu_{Y}\alpha-\theta^{\ast}\sigma^{2}-\alpha\mathbb{E}[Ye^{aD+\theta^{\ast}Y}]}{bc+\alpha e^{aD}\mathbb{E}[e^{\theta^{\ast}Y}]}\left(\sigma^{2}D(\theta^{\ast})-b+a\beta\mathbb{E}[e^{aD(\theta^{\ast})+\theta^{\ast}Y}]\right)\\
&\hspace{ 20em}=\beta\left(\mu_{Y}-\mathbb{E}[Ye^{aD(\theta^{\ast})+\theta^{\ast} Y}]\right)\\
&-bD(\theta^{\ast})+\frac{1}{2}\sigma^2D(\theta^{\ast})^2+\beta\left(\mathbb{E}[e^{aD(\theta^{\ast})+\theta^{\ast}Y}]-1\right)-\theta^{\ast}\mu_{Y}\beta=0.
\end{array}\right.
\end{equation}

Second, let us define the share measure $\bar{\mathbb{Q}}$ as
\begin{equation}
\frac{d\bar{\mathbb{Q}}}{d\mathbb{Q}}\bigg|_{\mathcal{F}_{t}}=\frac{S_{t}}{S_{0}}=e^{X_{t}}.
\end{equation}
Note that
\begin{align*}
\frac{S_{t}}{S_{0}}
&=e^{-\frac{1}{2}\sigma^{2}t+\sigma W_{t}^{\mathbb{Q}}-\mu_{Y}\int_{0}^{t}\lambda^{N}_{s}ds+\sum_{i=1}^{N_{t}}Y_{i}}
\\
&=e^{-\frac{1}{2}\sigma^{2}t+\sigma W_{t}^{\mathbb{Q}}}
\cdot\prod_{i=1}^{N_{t}}\frac{e^{Y_{i}}}{\mathbb{E}[e^{Y}]}
\cdot e^{\log\mathbb{E}[e^{Y}]N_{t}-\mu_{Y}\int_{0}^{t}\lambda^{N}_{s}ds}.
\end{align*}
Thus, under the share measure $\bar{\mathbb{Q}}$,
\begin{equation}
\bar{X}_{t}=\frac{1}{2}\sigma^{2}t+\sigma W_{t}^{\bar{\mathbb{Q}}}
-\mu_{Y}\int_{0}^{t}\bar{\lambda}_{s}^{\bar{N}}ds+\sum_{i=1}^{\bar{N}_{t}}\bar{Y}_{i},
\end{equation}
where $\bar{Y}_{i}$ are i.i.d. and according to $\bar{\mathbb{Q}}$
so that it has the probability distribution
\begin{equation*}
\frac{e^{Y}}{\mathbb{E}[e^{Y}]}d\mathbb{Q}
\end{equation*}
and $\bar{N}_{t}$ is an affine point process with intensity
\begin{equation*}
\bar{\lambda}^{\bar{N}}_{t}=\mathbb{E}[e^{Y}]\lambda^{N}_{t}.
\end{equation*}

Thus, $\bar{\mathbb{Q}}(\frac{1}{t}\bar{X}_{t}\in\cdot)$ satisfies
a large deviation principle with
\begin{equation*}
\bar{I}(x):=\sup_{\theta\in\mathbb{R}}\{\theta x-\bar{\Lambda}(\theta)\},
\end{equation*}
here
\begin{equation*}
\bar{\Lambda}(\theta):=\lim_{t\rightarrow\infty}\frac{1}{t}\log\mathbb{E}[e^{\theta\bar{X}_{t}}]
=\frac{1}{2}\sigma^2\theta^2+\left(\frac{1}{2}\sigma^2-\mu_{Y}\mathbb{E}[e^{Y}]\alpha\right)\theta+bc\bar{D}(\theta)
+\mathbb{E}[e^{Y}]\alpha\left(e^{a\bar{D}(\theta)}
\mathbb{E}[e^{\theta\bar{Y}}]-1\right),
\end{equation*}
where $\bar{D}(\theta)$ is the smaller solution of the equation
\begin{equation} \label{Dbartheta}
-b\bar{D}(\theta)+\frac{1}{2}\sigma^2\bar{D}(\theta)^2+\mathbb{E}[e^Y]\beta\left(\mathbb{E}[e^{a\bar{D}(\theta)+\theta\bar{Y}}]-1\right)
-\theta\mu_{Y}\mathbb{E}[e^{Y}]\beta=0.
\end{equation}

As a corollary, $\bar{\mathbb{Q}}(-\frac{1}{t}\bar{X}_{t}\in\cdot)$ satisfies
a large deviation principle with the rate function $\bar{I}(-x)$. Moreover, for any $x\in\mathbb{R}$ and for any sufficiently small $\delta>0$,
\begin{equation*}
\bar{\mathbb{Q}}\left(x-\delta<\frac{\bar{X}_{t}}{t}<x+\delta\right)
=\mathbb{E}\left[e^{X_{t}}1_{x-\delta<\frac{X_{t}}{t}<x+\delta}\right],
\end{equation*}
which implies that
\begin{equation*}
\bar{I}(x)=I(x)-x.
\end{equation*}

Third, following the similar lines in Corollary 2.4 in \cite{FordeJacquier2011}, we have
\begin{equation} \label{eqn:large-strike-large-maturity-rate-function}
I(x)-x
=
\begin{cases}
-\lim_{T\rightarrow\infty}\frac{1}{T}\log\mathbb{E}[(S_{T}-S_{0}e^{xT})^{+}] &\text{for $x\geq x_{R}$},
\\
-\lim_{T\rightarrow\infty}\frac{1}{T}\log(S_{0}-\mathbb{E}[(S_{T}-S_{0}e^{xT})^{+}]) &\text{for $x_{L}\leq x\leq x_{R}$},
\\
-\lim_{T\rightarrow\infty}\frac{1}{T}\log\mathbb{E}[(S_{0}e^{xT}-S_{T})^{+}] &\text{for $x\leq x_{L}$},
\end{cases}
\end{equation}
from which we can compute that
\begin{equation}
x_{L}=\Lambda'(0),
\qquad
x_{R}=\bar{\Lambda}'(0).
\end{equation}
Differentiating $\Lambda(\theta)$ with respect to $\theta$, we get
\begin{equation}\label{Lambda}
\Lambda'(\theta)=\sigma^2\theta-\left(\frac{1}{2}\sigma^2+\mu_{Y}\alpha\right)+bcD'(\theta)+\alpha e^{aD(\theta)}\left( aD'(\theta)\mathbb{E}[e^{\theta Y}] + \mathbb{E}[Ye^{\theta Y}] \right).
\end{equation}
From equation (\ref{Dtheta}), we have
\begin{equation*}
D'(\theta)=\frac{\beta\left(\mu_{Y}-\mathbb{E}[Ye^{aD(\theta)+\theta Y}]\right)}{\sigma^{2}D(\theta)-b+a\beta\mathbb{E}[e^{aD(\theta)+\theta Y}]},
\end{equation*}
and $D(0)=0$ from $\mathbb{E}[e^{aD}]=1$, so
\begin{equation}\label{D0}
D'(0)=\frac{\beta\left(\mu_{Y}-\mathbb{E}[Y]\right)}{a\beta-b}.
\end{equation}
Plugging equation (\ref{D0}) into equation (\ref{Lambda}), we have
\begin{equation*}
x_{L}=\Lambda'(0)=-\left(\frac{1}{2}\sigma^2+\mu_{Y}\alpha\right)+\left(bc+a\alpha\right)\frac{\beta\left(\mu_{Y}-
\mathbb{E}[Y]\right)}{a\beta-b} + \alpha \mathbb{E}[Y].
\end{equation*}
Similarly, differentiating $\bar{\Lambda}(\theta)$ w.r.t. $\theta$,
\begin{equation}\label{barLambda}
\bar{\Lambda}'(\theta)=\sigma^2\theta+\left(\frac{1}{2}\sigma^2-\mu_{Y}\mathbb{E}[e^{Y}]\alpha\right)+bc\bar{D}'(\theta)+\mathbb{E}[e^Y]\alpha e^{a \bar{D}(\theta)}\left(a \bar{D}'(\theta) \mathbb{E}[e^{\theta \bar{Y}}] + \mathbb{E}[\bar{Y} e^{\theta \bar{Y}}]\right).
\end{equation}
In addition, from equation (\ref{Dbartheta}) we have
\begin{equation*}
\bar{D}'(\theta)=\frac{\beta\mathbb{E}[e^{Y}]\left(\mu_{Y}-\mathbb{E}[\bar{Y}e^{a\bar{D}(\theta)+\theta\bar{Y}}]\right)}{\sigma^{2}
\bar{D}(\theta)-b+a\beta\mathbb{E}[e^{Y}]\mathbb{E}[e^{a\bar{D}(\theta)+\theta\bar{Y}}]},
\end{equation*}
and $\bar{D}(0)=0$ from $\mathbb{E}[e^{a\bar{D}}]=1$, so
\begin{equation} \label{barD0}
\bar{D}'(0)=\frac{\beta\mathbb{E}[e^{Y}]\left(\mu_{Y}-\mathbb{E}[\bar{Y}]\right)}{a\beta\mathbb{E}[e^{Y}]-b}.
\end{equation}
Plugging equation (\ref{barD0}) into equation (\ref{barLambda}), we have
\begin{equation*}
x_{R}=\bar{\Lambda}'(0)=\left(\frac{1}{2}\sigma^2-\mu_{Y}\mathbb{E}[e^{Y}]\alpha\right)+\left(bc+a\mathbb{E}[e^Y]\alpha
\right)\frac{\mathbb{E}[e^{Y}]\beta\left(\mu_{Y}-\mathbb{E}[\bar{Y}]\right)}{a\mathbb{E}[e^{Y}]\beta-b} + \mathbb{E}[e^Y] \alpha \mathbb{E}[\bar{Y}].
\end{equation*}

In summary,
\begin{equation*}
x_{L}=\Lambda'(0)=-\left(\frac{1}{2}\sigma^2+\mu_{Y}\alpha\right)+\left(bc+a\alpha\right)\frac{\beta\left(\mu_{Y}-
\mathbb{E}[Y]\right)}{a\beta-b} + \alpha \mathbb{E}[Y]
\end{equation*}
and
\begin{equation*}
x_{R}=\bar{\Lambda}'(0)=\left(\frac{1}{2}\sigma^2-\mu_{Y}\mathbb{E}[e^{Y}]\alpha\right)+\left(bc+a\mathbb{E}[e^Y]\alpha
\right)\frac{\mathbb{E}[e^{Y}]\beta\left(\mu_{Y}-\mathbb{E}[\bar{Y}]\right)}{a\mathbb{E}[e^{Y}]\beta-b} + \mathbb{E}[e^Y] \alpha \mathbb{E}[\bar{Y}].
\end{equation*}

Fourth, it follows from Corollary 2.14 in \cite{FordeJacquier2011} that in the joint regime of large-maturity, large-strike with $k=\log(K/S_{0})$
$(T\rightarrow\infty$, $|k|\rightarrow\infty$), the implied volatility $\sigma_{\text{BS}}(k,T)$
approaches the limit
\begin{equation*}
\lim_{T\rightarrow\infty}\sigma^{2}_{\text{BS}}(xT,T)=\sigma_{\infty}^{2}(x),
\end{equation*}
where
\begin{equation*}
\sigma_{\infty}^{2}(x)=
\begin{cases}
2(2I(x)-x-2\sqrt{I^{2}(x)-xI(x)}) & x\in(-\infty,x_{L})\cup(x_{R},\infty)
\\
2(2I(x)-x+2\sqrt{I^{2}(x)-xI(x)}) & x\in[x_{L},x_{R}]
\end{cases}.
\end{equation*}
\end{proof}
\subsection{Asymptotics of implied volatility in fixed-maturity, large-strike and small-strike regimes}

${}\quad$ In this section, we apply Lee's moment formula (\cite{Lee2004}) to derive the asymptotics
for the Black-Scholes implied volatility in fixed-maturity, large-strike ($K\rightarrow\infty$) and small-strike ($K\rightarrow 0$) regimes.

Define
\begin{equation}\label{p tilde}
\tilde{p}:=\sup\left\{p:\mathbb{E}^{\mathbb{Q}}[S_{T}^{1+p}]<\infty\right\},
\end{equation}
and
\begin{equation}\label{q tilde}
\tilde{q}:=\sup\left\{q:\mathbb{E}^{\mathbb{Q}}[S_{T}^{-q}]<\infty\right\}.
\end{equation}

The following lemma gives an explicit formula relating the right-hand (or large-$K$ or positive-$x$) tail slope and the left-hand (or small-$K$ or negative-$x$) tail slope to how many finite moments the
underlying possesses.

\begin{lemma}\label{moment formula}(\cite{Lee2004}) For $k=\log(K/S_{0})$. Let $\beta_R:=\mathop{\limsup}\limits_{k\rightarrow+\infty}\frac{\sigma^{2}_{\text{BS}}(k)}{|k|/T}$ and $\beta_L:=\mathop{\limsup}\limits_{k\rightarrow-\infty}\frac{\sigma^{2}_{\text{BS}}(k)}{|k|/T}$.
Then $\beta_R\in[0,2]$ and $\beta_L\in[0,2]$ and
\begin{align}
&\tilde{p}=\frac{1}{2\beta_R}+\frac{\beta_R}{8}-\frac{1}{2}\nonumber,\\
&\tilde{q}=\frac{1}{2\beta_L}+\frac{\beta_L}{8}-\frac{1}{2}\nonumber,
\end{align}
where $\frac{1}{0}:=\infty$. Equivalently,
\begin{align}
&\beta_R=2-4(\sqrt{\tilde{p}^2+\tilde{p}}-\tilde{p})\nonumber,\\
&\beta_L=2-4(\sqrt{\tilde{q}^2+\tilde{q}}-\tilde{q})\nonumber,
\end{align}
where the right-hand expression is to be read as zero, in the case $\tilde{p}=\infty$ or $\tilde{q}=\infty$.
\end{lemma}

Applying Lee's moment formula, we obtain the following results for our model:
\begin{theorem} \label{fmlsss} In the joint regime of fixed-maturity, large-strike (small-strike) with $k=\log(K/S_{0})$
($|k|\rightarrow\infty)$, the implied volatility $\sigma_{\text{BS}}(k,T)$
approaches the limit
\begin{equation}\label{eqn:joint regime of fixed T}
\begin{aligned}
\mathop{\limsup}\limits_{k\rightarrow+\infty}\frac{\sigma^{2}_{\text{BS}}(k,T)}{|k|/T}
&=2-4(\sqrt{\tilde{p}^2+\tilde{p}}-\tilde{p}), \quad ({\rm large~strike}),\\
\mathop{\limsup}\limits_{k\rightarrow-\infty}\frac{\sigma^{2}_{\text{BS}}(k,T)}{|k|/T}
&=2-4(\sqrt{\tilde{q}^2+\tilde{q}}-\tilde{q}),  \quad ({\rm small~strike}),
\end{aligned}
\end{equation}
where $\tilde{p}$ and $\tilde{q}$ are defined via
\begin{equation*}
\int_0^\infty\frac{d\bar{D}}{H(\bar{D};\tilde{p}-1)}=T, \qquad
\int_0^\infty\frac{d\bar{D}}{H(\bar{D};-\tilde{q})}=T,
\end{equation*}
and
$$H(\bar{D};p):=-b\bar{D}+\frac{1}{2}\sigma^2\bar{D}^2+\beta\int_{\mathbb{R}}(e^{\bar{D}a+p y}-1)Q(dy)-p\mu_Y\beta.$$
\end{theorem}


\begin{proof} Let us determine the $\tilde{p}$ and $\tilde{q}$ in (\ref{p tilde}) and (\ref{q tilde}) for $S_T$ in (\ref{affine point dynamics}). Recall that $\tilde{p}+1$ is the largest $p$ such that $\mathbb{E}[e^{pX_{T}}]<\infty$. From (\ref{moment gen_X}), we know
$$\mathbb{E}[e^{pX_{T}}]=e^{(-\frac{1}{2}p\sigma^{2}+\frac{1}{2}p^{2}\sigma^{2}-p\mu_Y\alpha)T
+\bar{D}(T;p)\lambda+pL+\bar{F}(T;p)},$$
where $\bar{D}(T;p)$ and $\bar{F}(T;p)$ solve a set of ODEs. According to the ODEs (\ref{DF}), we see $\bar{F}(T;p)$ is determined by $\bar{D}(T;p)$, so $\mathbb{E}[e^{pX_{T}}]<\infty \Longleftrightarrow \bar{D}(T;p)<\infty $ and the critical $\tilde{p}$ is the value of $p$ such that $\bar{D}(T;p)=\infty$.
Recall that $\bar{D}(t;p)$ solves the ODE in (\ref{DF})
\begin{equation}\label{D}
\left\{\begin{array}{r@{}l@{\qquad}l}
&\bar{D}'(t;p)=-b\bar{D}(t;p)+\frac{1}{2}\sigma^2\bar{D}^2(t;p)+\beta\int_{\mathbb{R}}(e^{\bar{D}(t;p)a+p y}-1)Q(dy)-p\mu_Y\beta:=H(\bar{D};p),\\
& \bar{D}(0;p)=0.
\end{array}\right.
\end{equation}
Define $\bar{D}'(t;p)=H(\bar{D};p)$,
\begin{equation}\label{intergral p}
\int_{\bar{D}(0;p)}^{\bar{D}(T;p)}\frac{d\bar{D}}{H(\bar{D};p)}=\int_0^Tdt=T.
\end{equation}
Therefore the critical $p=\tilde{p}-1$ satisfies
$\int_{0}^{\infty}d\bar{D}/H(\bar{D},p)=T$ as $\bar{D}(T;p)=\infty$.
For a given maturity $T$, we can find a $p$ which satisfies 
\begin{equation}\label{eqn:fixed_maturity_T_vs_p}
\int_0^\infty\frac{dx}{-bx+\frac{1}{2}\sigma^2x^2+\beta e^{ax}\mathbb{E}[e^{pY}]-\beta-p\mu_Y\beta}=T.
\end{equation}
Similarly, the critical $\tilde{q}=-q$ satisfies $\int_{0}^{\infty}d\bar{D}/H(\bar{D},q)=T$.
\end{proof}

\begin{remark}Numerical examples are provided in later sections to verify the existence of $p$ and $q$ values for different $T$'s in (\ref{eqn:fixed_maturity_T_vs_p}).
\end{remark}

\section{Numerical study}
${}\quad$ In this section, we provide some numerical study results.
The strength of the self-exciting process is controlled by $a$ in (\ref{APP}) and $\beta$ in the intensity function $\lambda^N_t$. Hence we vary $a$ and $\beta$ values to study how these two parameters affect the rate function and the asymptotic implied volatility. $a$ is chosen to be $0.05,~0.5$ and $1$ and $\beta$ is chosen to be $0.1,~0.25$ and $0.5$. For all numerical studies, we define the jump size $Y \sim \mathcal{N}(0,\sigma^{2})$. Other parameters are $b=1$, $c=0.05$, $\alpha=1$, $\sigma^2=0.1$ and $\delta^2=0.1$.

Figure \ref{fig:I_x_different_a} shows the rate function for selected $a$ values. One should notice as $a$ increases, the growth rate of $I(x)$ increases. This is expected as more rare events occur when $a$ increases, so the rate function $I(x)$ tends to be smaller. The right figure is the zoom-in of the left figure and it shows the minimums do not coincide. Rate function $\bar{I}(x)$ is shown in Figure \ref{fig:I_bar_x_different_a} and it has similar behaviors as $I(x)$ in Figure \ref{fig:I_x_different_a}. Figure \ref{fig:implied_vol_x_different_a} shows the asymptotic of implied volatility in the large-maturity and large-strike regime for different $a$ values. The affine point jump-diffusion model can capture the implied volatility smiles in this regime. \cite{FordeJacquier2011} found similar implied volatility smiles for the Heston model in the same regime. Consider the At-The-Money cases when $x=0$, the ATM volatility increases as $a$ increases. It is because as more rare events occur, the implied volatility is higher. Besides, the growth rate of the implied volatility into In-The-Money/Out-The-Money increases as $a$ increases.

Numerical results for different $\beta$ values are shown in Figures \ref{fig:I_x_different_beta}, \ref{fig:I_bar_x_different_beta} and \ref{fig:implied_vol_x_different_beta}. Because the parameter $\beta$ controls the strength of the self-exciting process intensity, so varying $\beta$ has similar effects as varying $a$.

\begin{figure}[H]
	\centering 
	\subfigure{
	\includegraphics[scale=0.35]{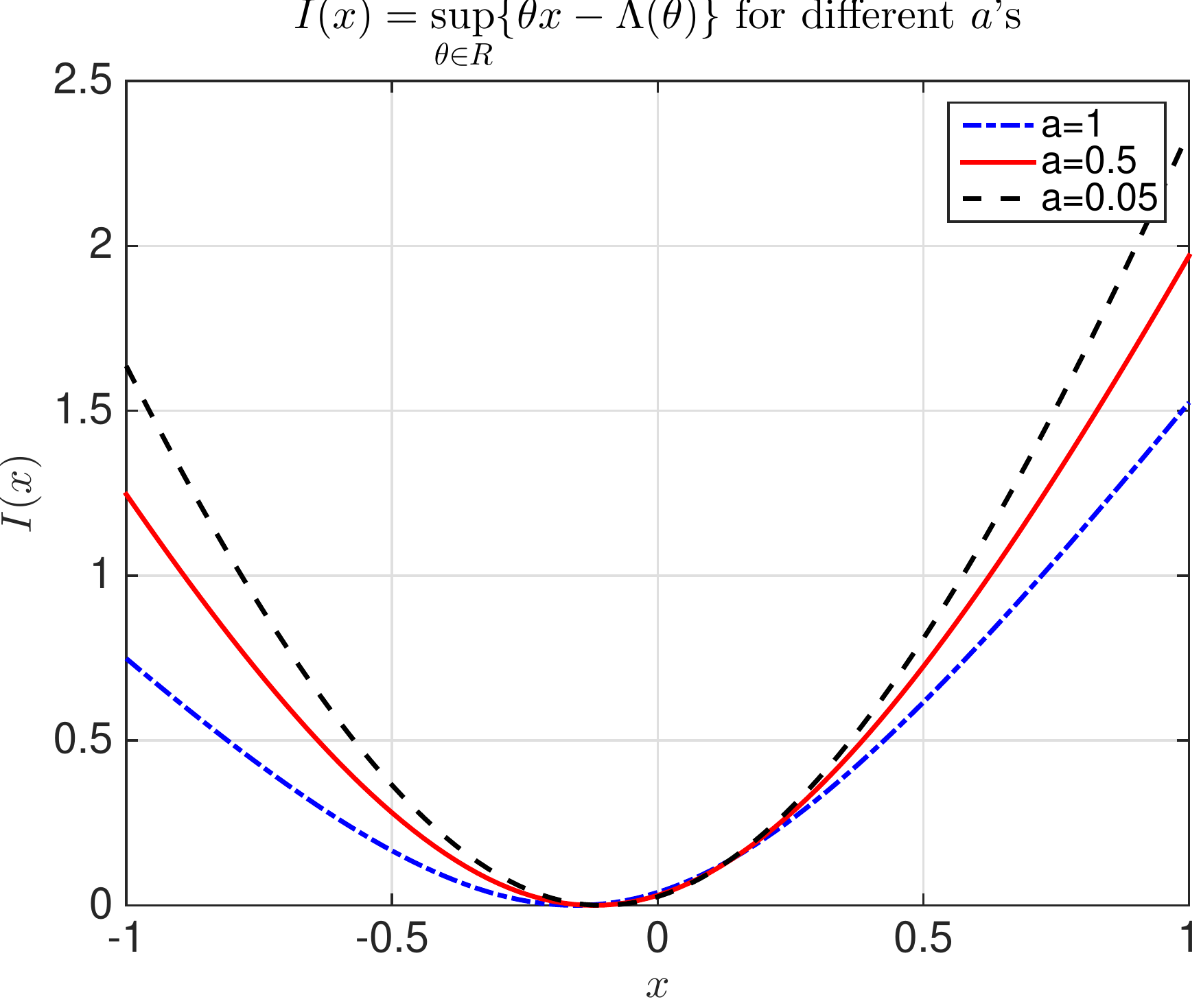}
	} 
	\subfigure{ 
	\includegraphics[scale=0.35]{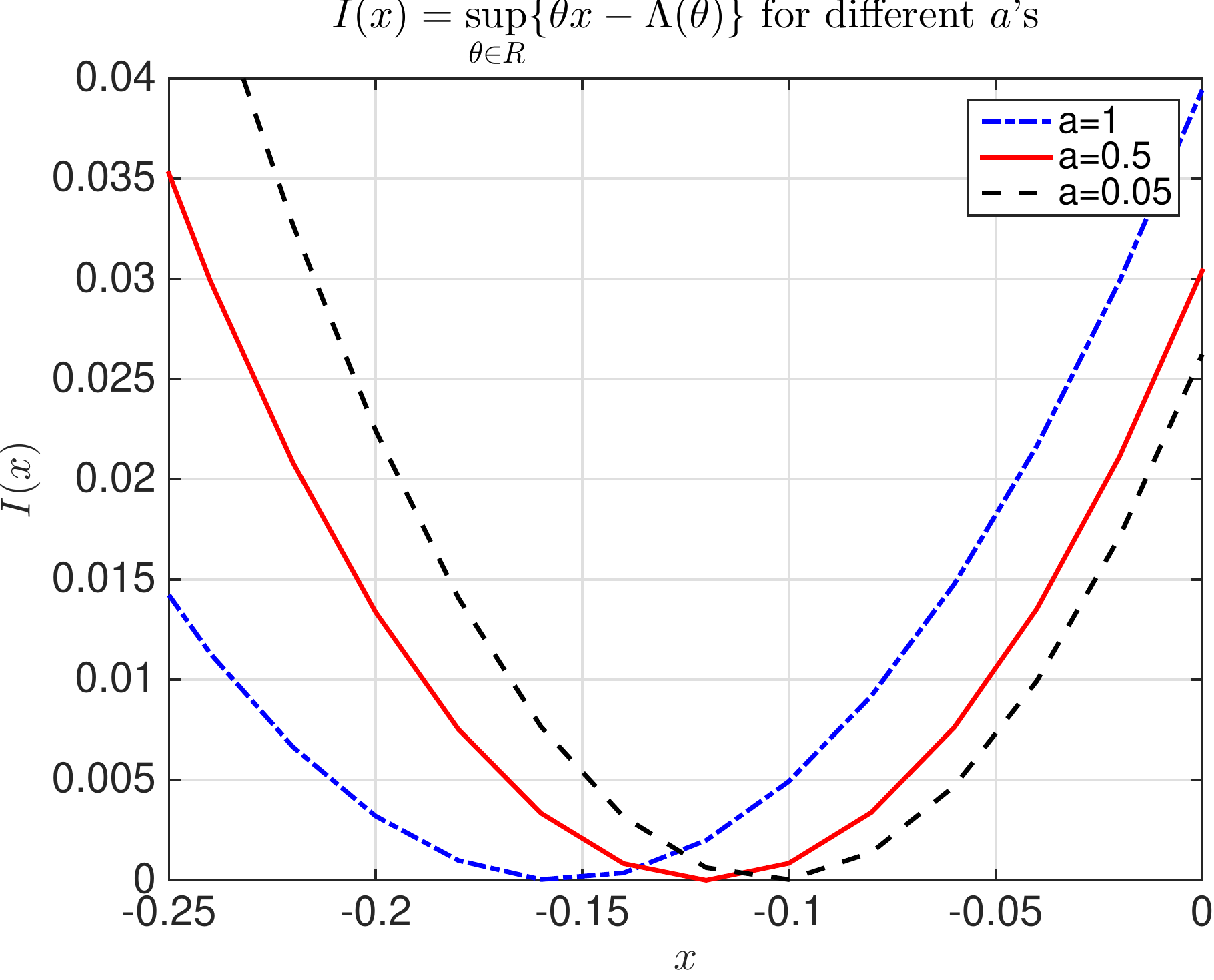}
	} 
	\caption{Left: $I(x)$ for $a=0.05,~0.5$ and $1$; Right: Zoom-in of left figure near $I(x)=0$.}
	\label{fig:I_x_different_a}
\end{figure}

\begin{figure}[H]
	\centering 
	\subfigure{
	\includegraphics[scale=0.35]{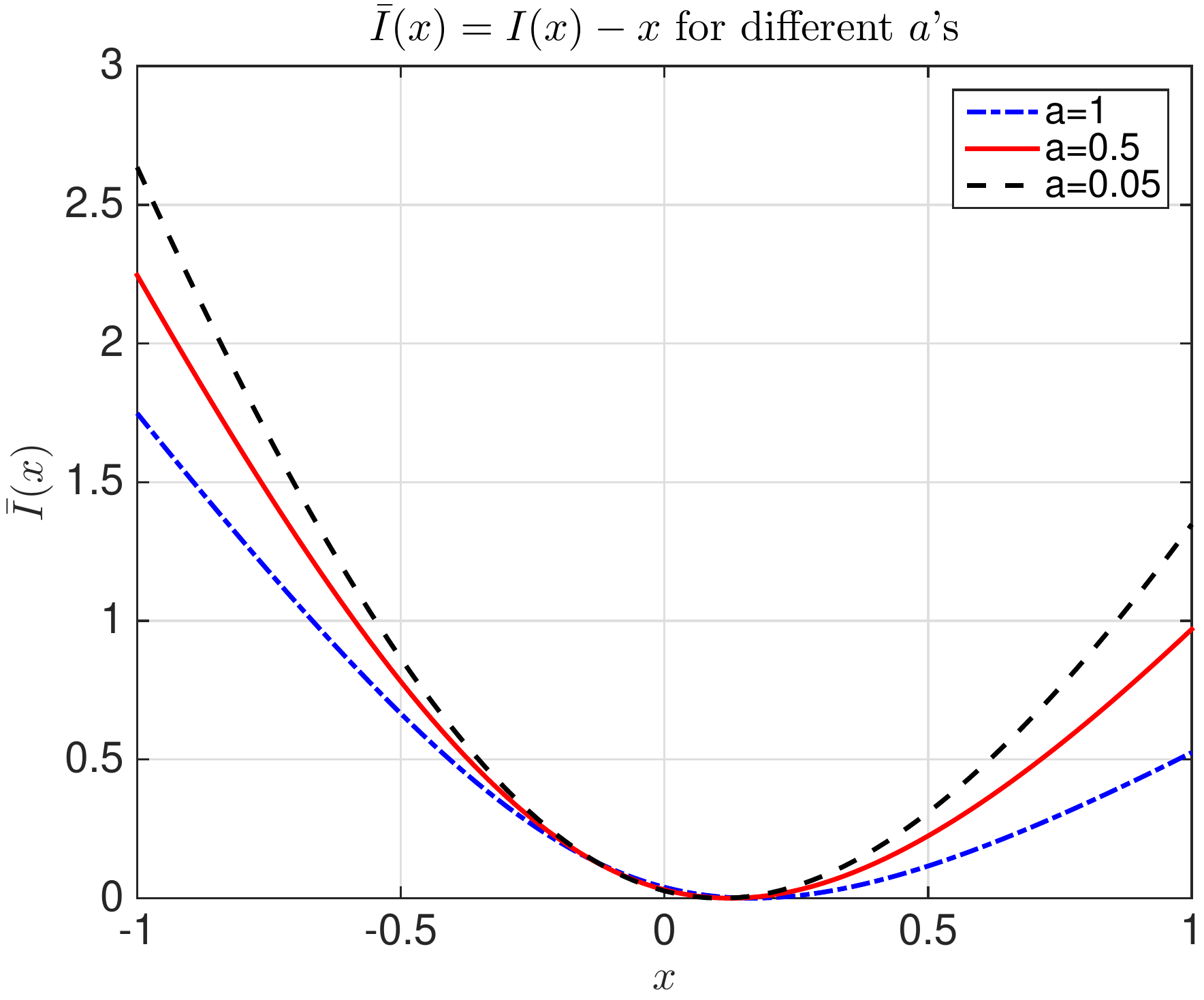}
	} 
	\subfigure{ 
	\includegraphics[scale=0.35]{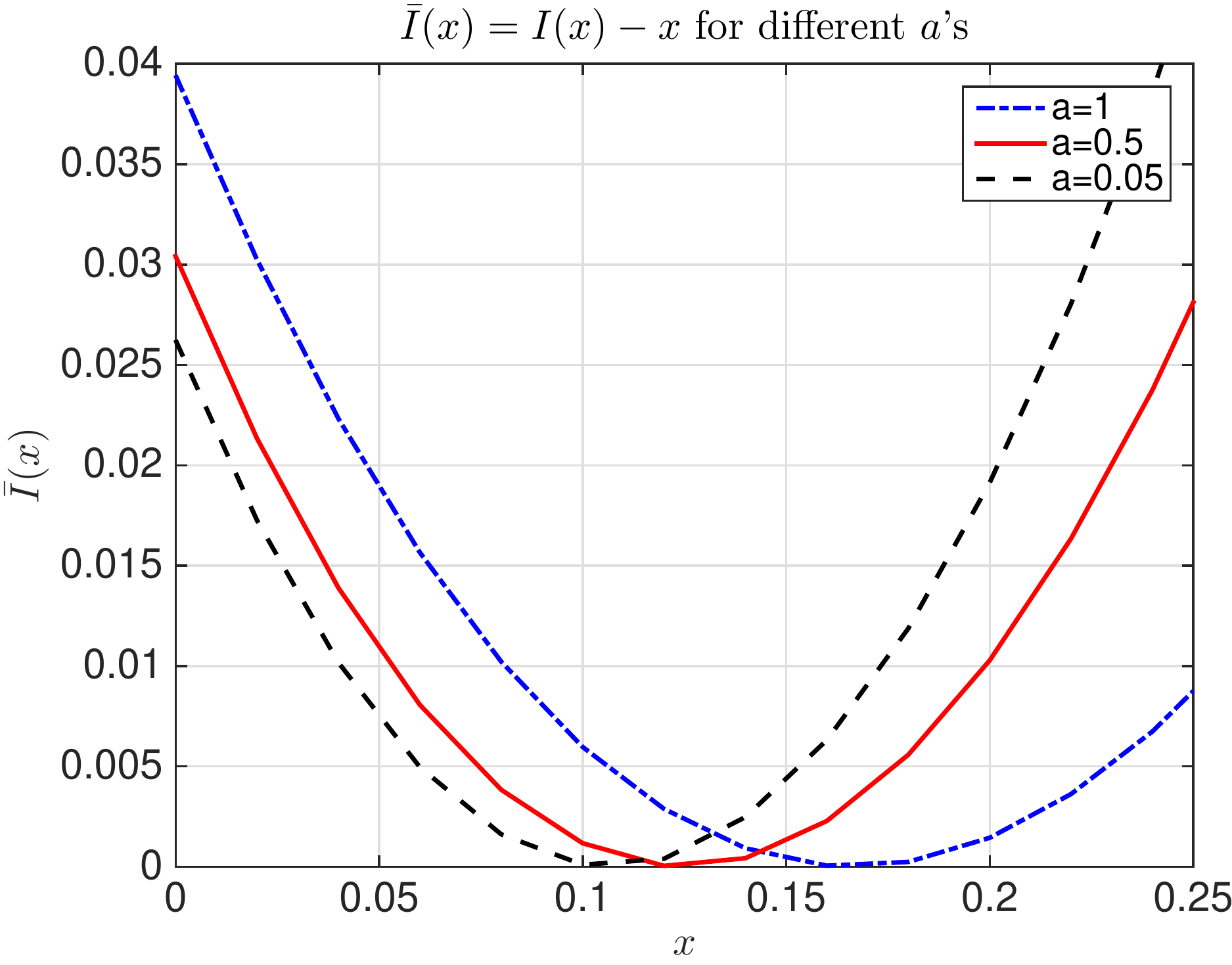}
	} 
	\caption{Left: $\bar{I}(x)$ for $a=0.05,~0.5$ and $1$; Right: Zoom-in of left figure near $\bar{I}(x)=0$.}
	\label{fig:I_bar_x_different_a}
\end{figure}

\begin{figure}[H]
	\centering
	\includegraphics[scale=0.5]{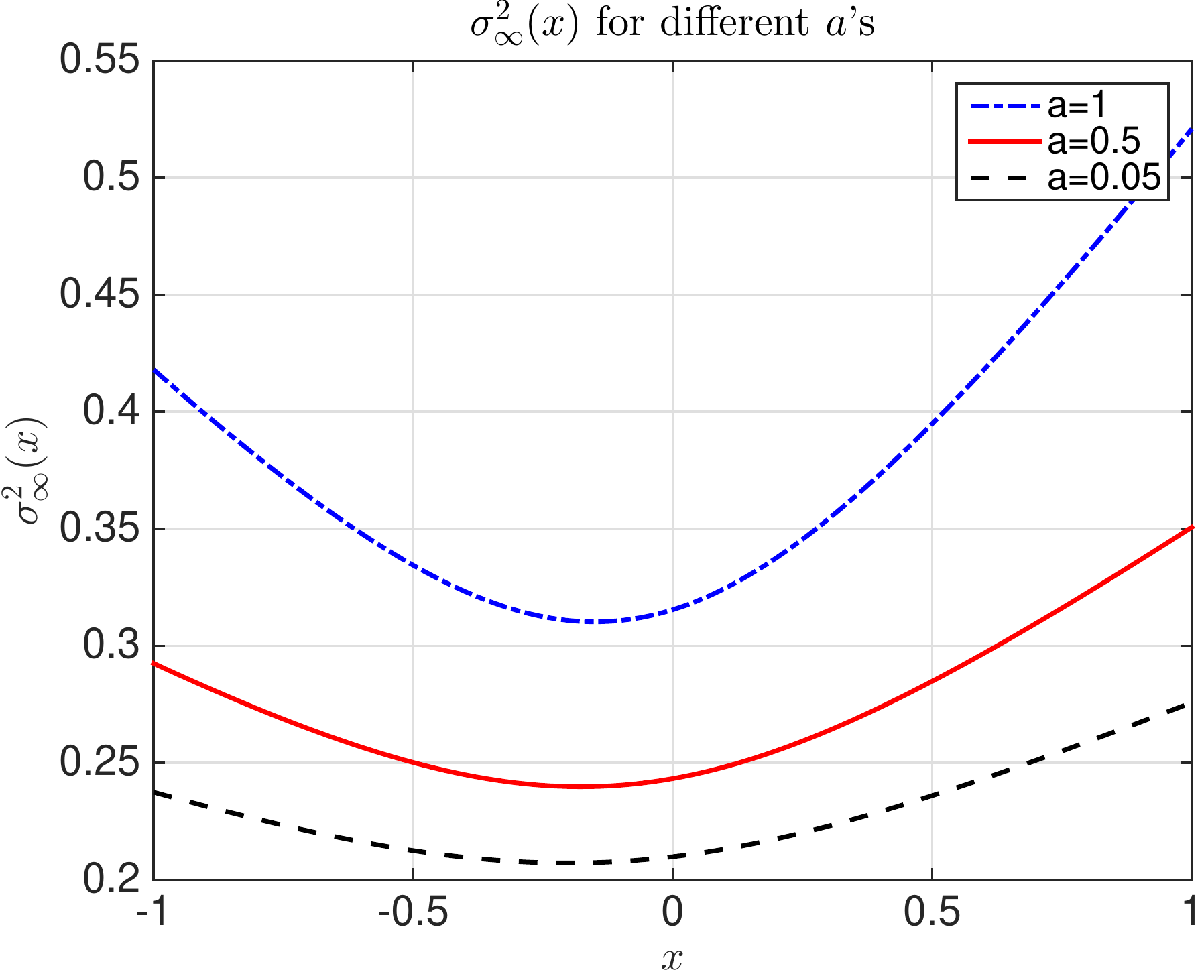}\\
	\caption{$\sigma_\infty^2(x)$ for $a=0.05,~0.5$ and $1$.}
	\label{fig:implied_vol_x_different_a}
\end{figure}

\begin{figure}[H]
	\centering 
	\subfigure{
	\includegraphics[scale=0.35]{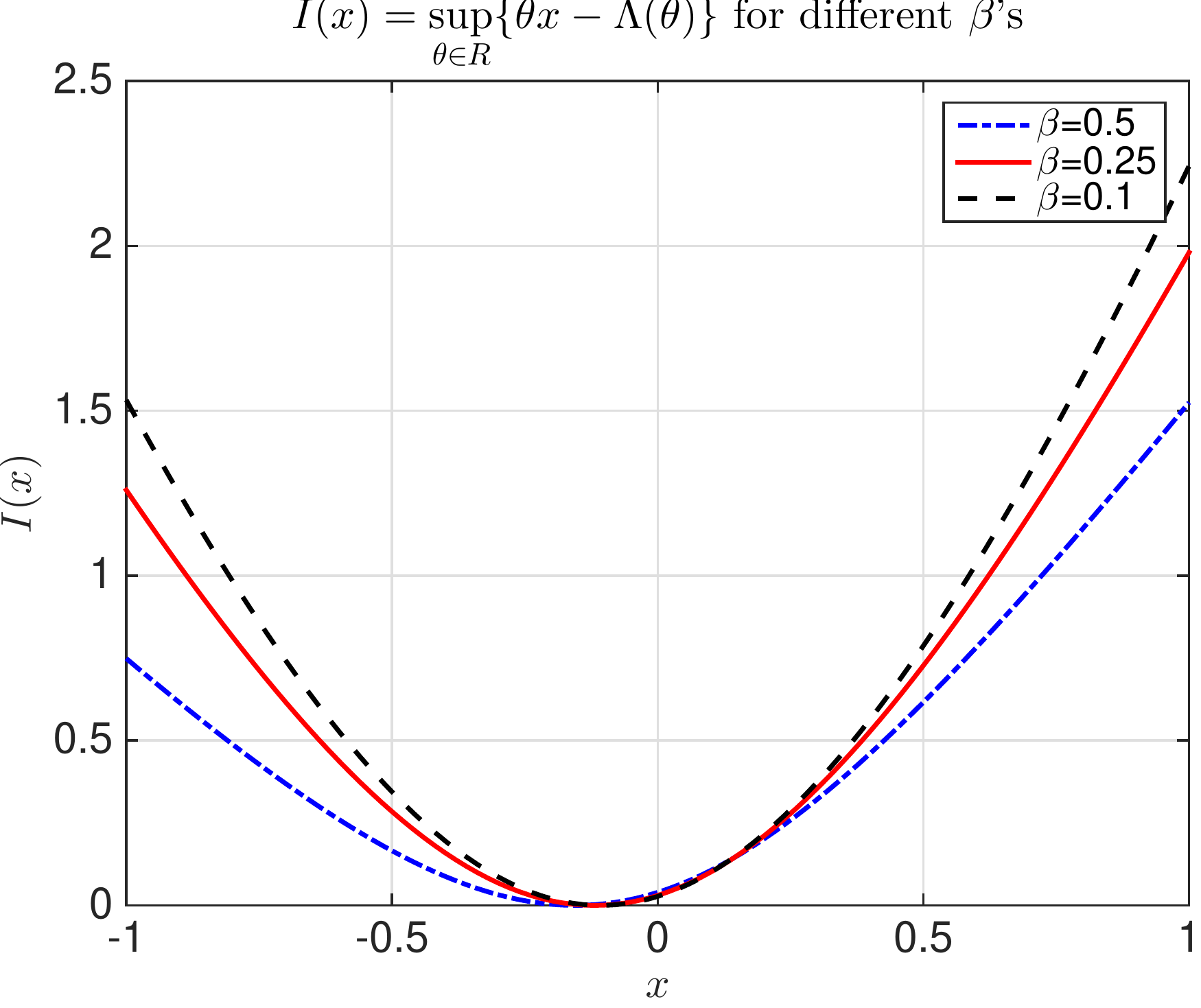}
	} 
	\subfigure{ 
	\includegraphics[scale=0.35]{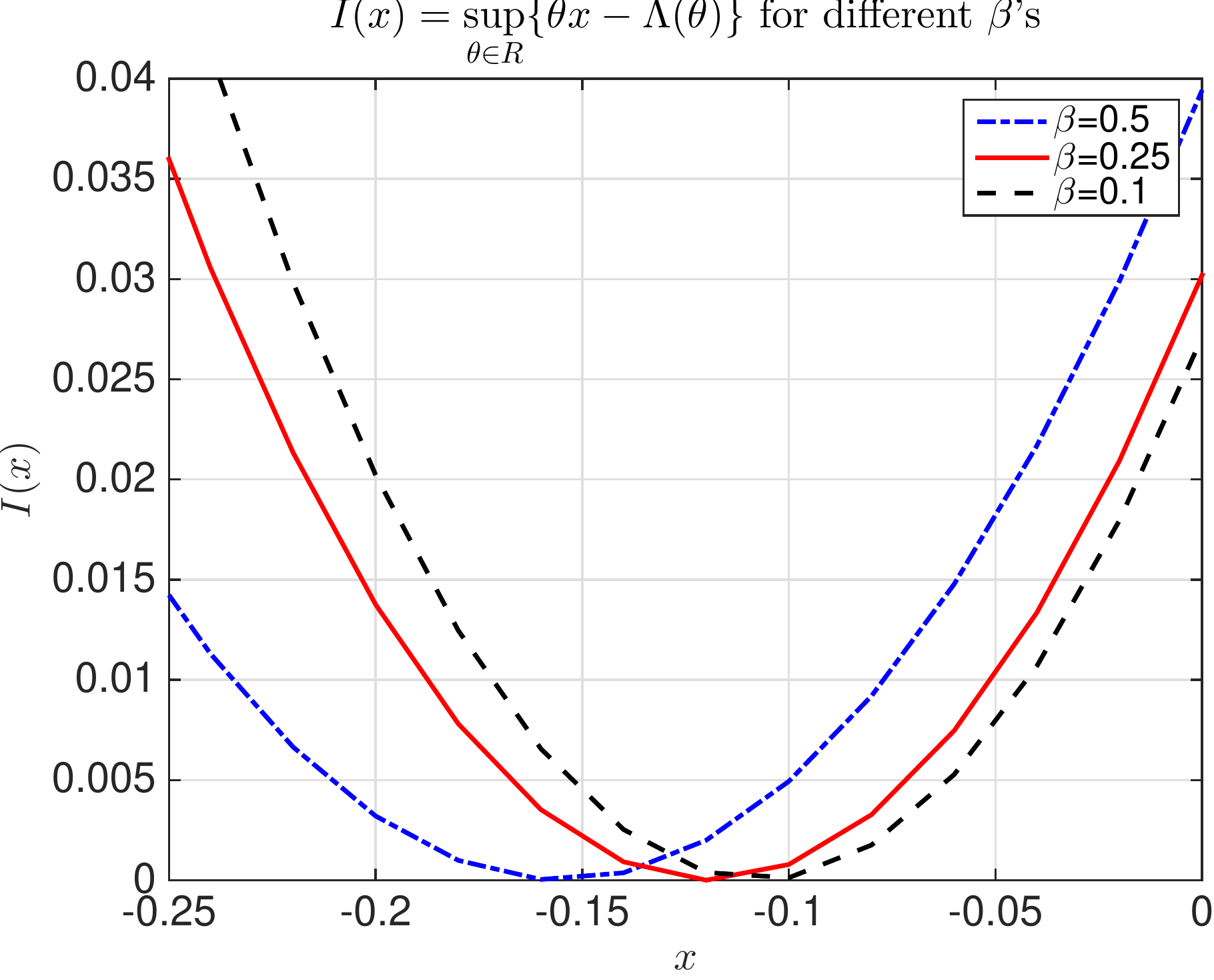}
	} 
	\caption{Left: $I(x)$ for $\beta=0.1,~0.25$ and $0.5$; Right: Zoom-in of left figure near $I(x)=0$.}
	\label{fig:I_x_different_beta}
\end{figure}

\begin{figure}[H]
	\centering 
	\subfigure{
	\includegraphics[scale=0.35]{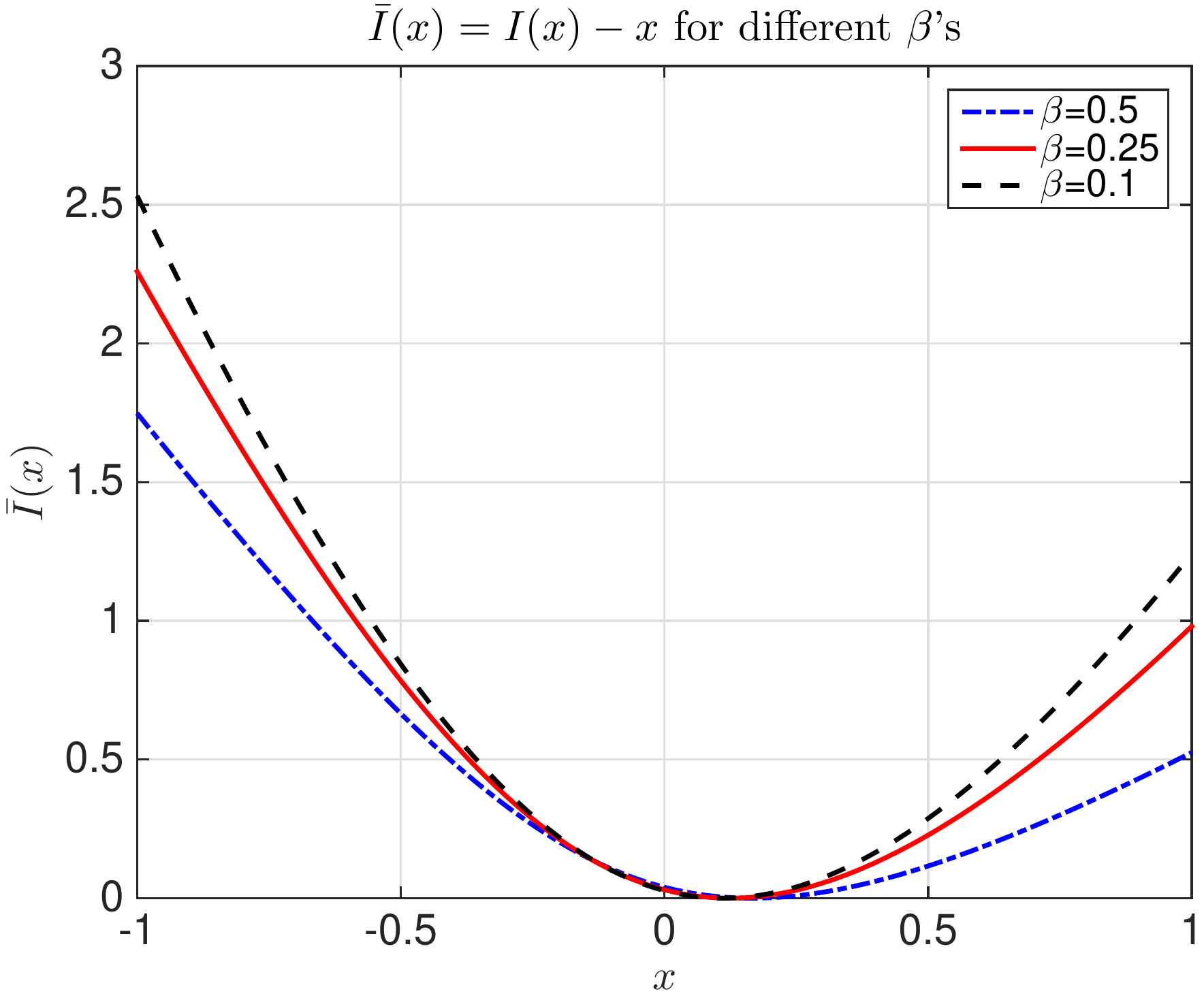}
	} 
	\subfigure{ 
	\includegraphics[scale=0.35]{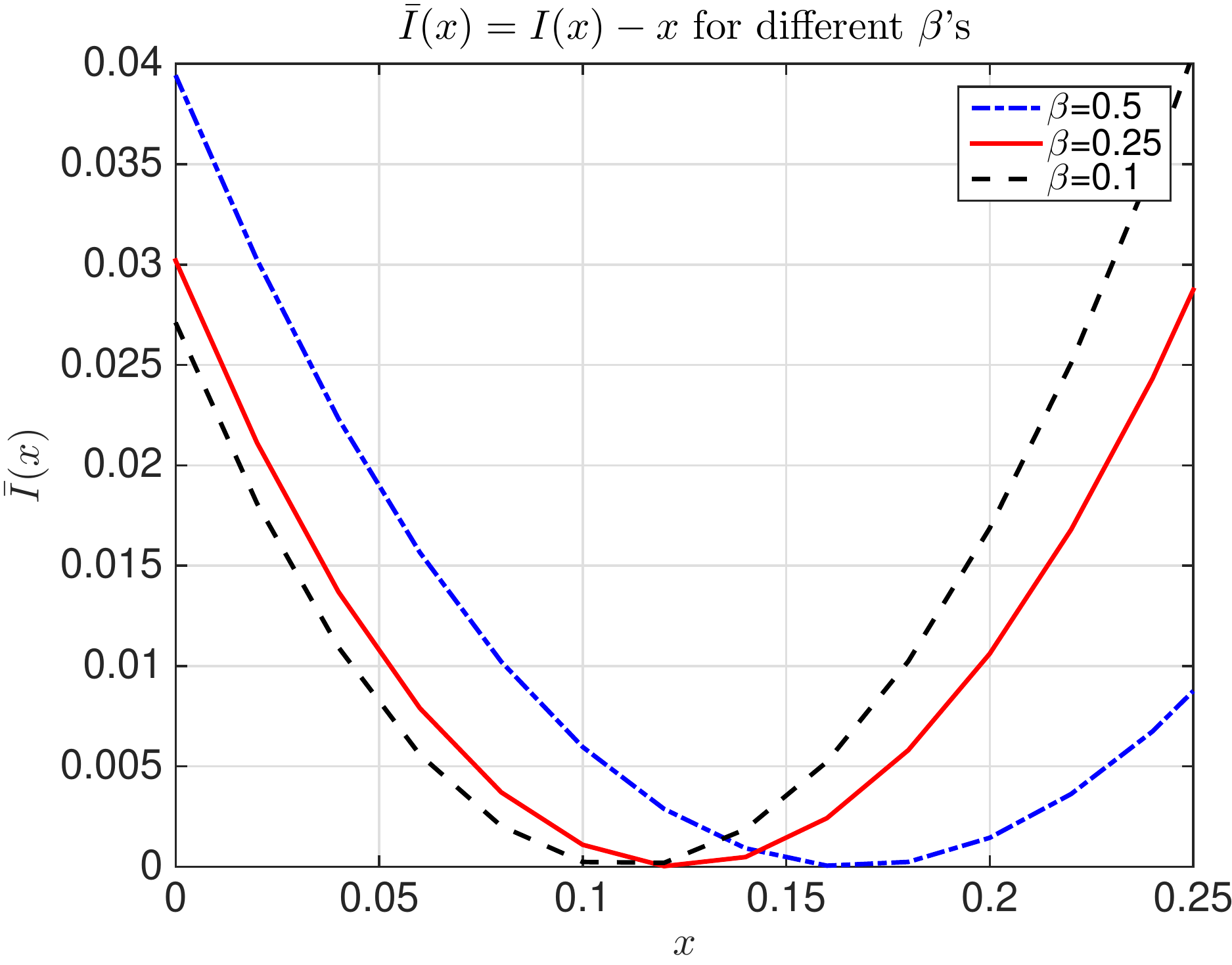}
	} 
	\caption{Left: $\bar{I}(x)$ for $\beta=0.1,~0.25$ and $0.5$; Right: Zoom-in of left figure near $\bar{I}(x)=0$.}
	\label{fig:I_bar_x_different_beta}
\end{figure}

\begin{figure}[H]
	\centering
	\includegraphics[scale=0.5]{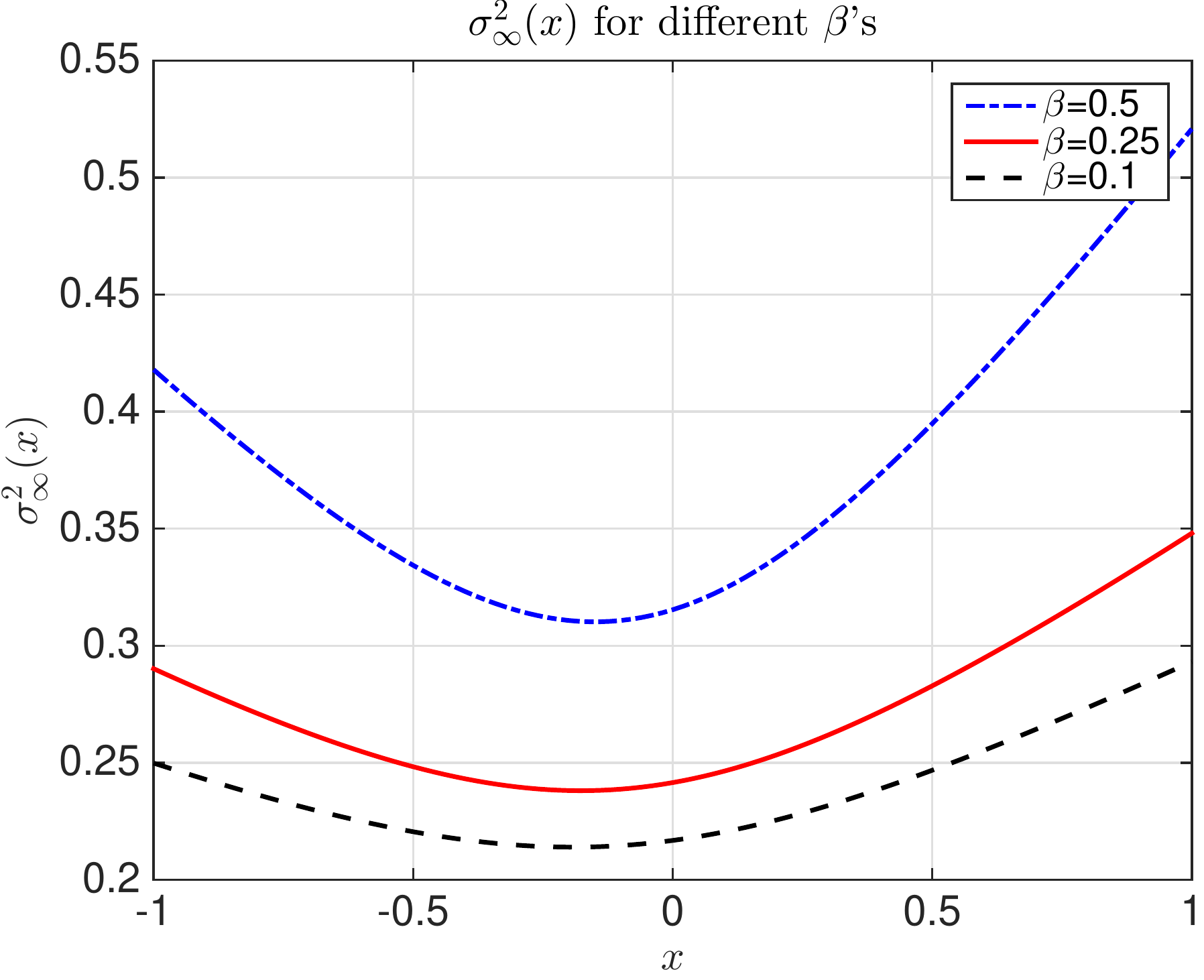}\\
	\caption{$\sigma_\infty^2(x)$ for $\beta=0.1,~0.25$ and $0.5$.}
	\label{fig:implied_vol_x_different_beta}
\end{figure}

Numerical examples in fixed-maturity large, small-strike and large-strike are presented. 
The left figure in Figure \ref{fig:fixed_T_small_large_strike_a} shows the ratio of Black-Scholes implied volatility to log-moneyness in the fixed-maturity and large-strike regime for different $a$ values; while right figure displays the ratio in the fixed-maturity and small-strike regime.
The maturity $T$ is chosen within a reasonable range.
In both figures, we observe that, for a given $T$, the ratio of implied volatility to log-moneyness increases as the self-exciting intensity parameter $a$ increases.
It is interesting to point out that, in these regimes, the ratio of Black-Scholes implied volatility to log-moneyness decreases as maturity increases. This is practically observed on an implied volatility surface.
Results for various values of $\beta$'s are provided in Figure \ref{fig:fixed_T_small_large_strike_beta}. We obtain similar results because $\beta$ controls the strength of the self-exciting process as well.

\begin{figure}[H]
	\centering 
	\subfigure{
	\includegraphics[scale=0.35]{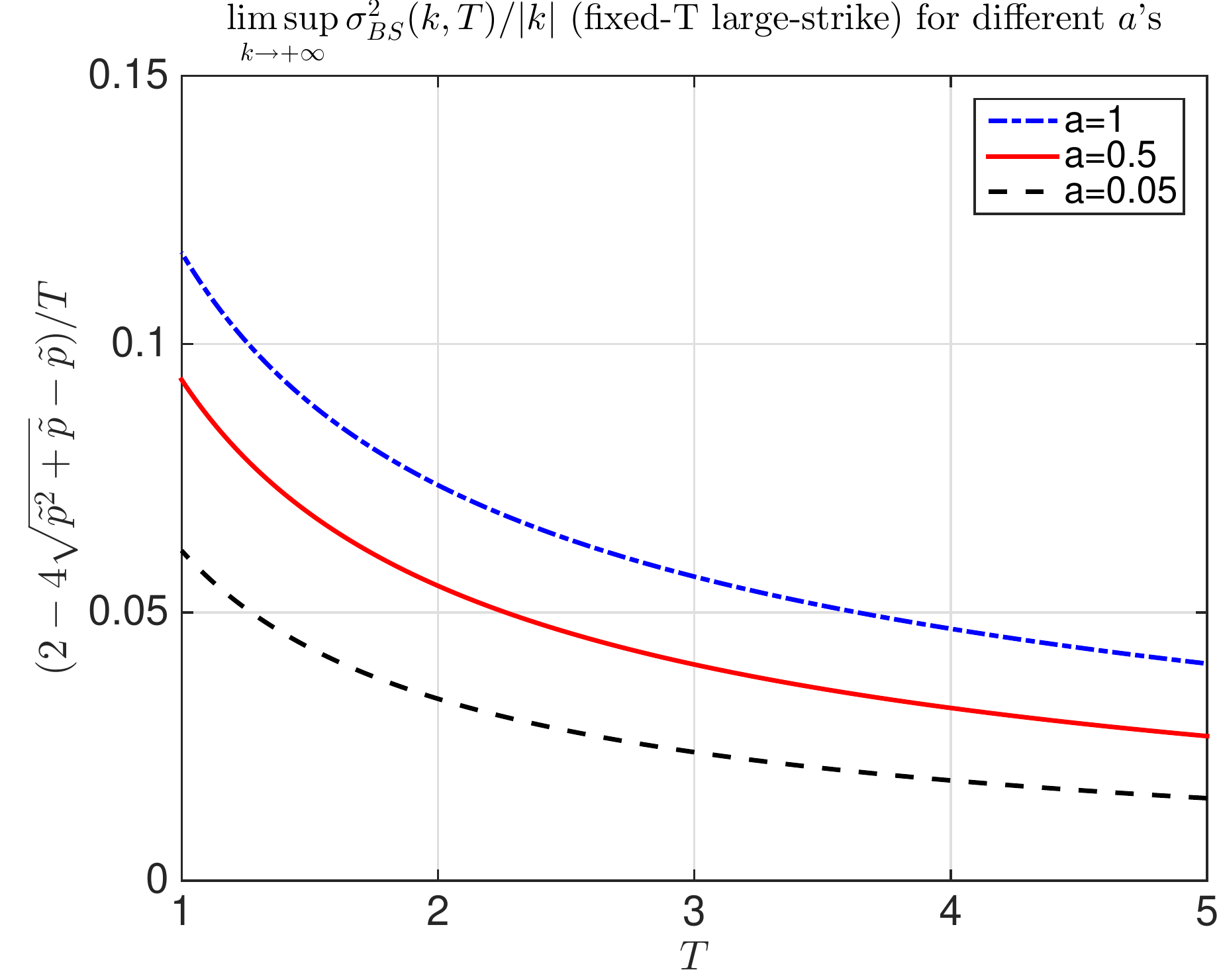}
	} 
	\subfigure{ 
	\includegraphics[scale=0.35]{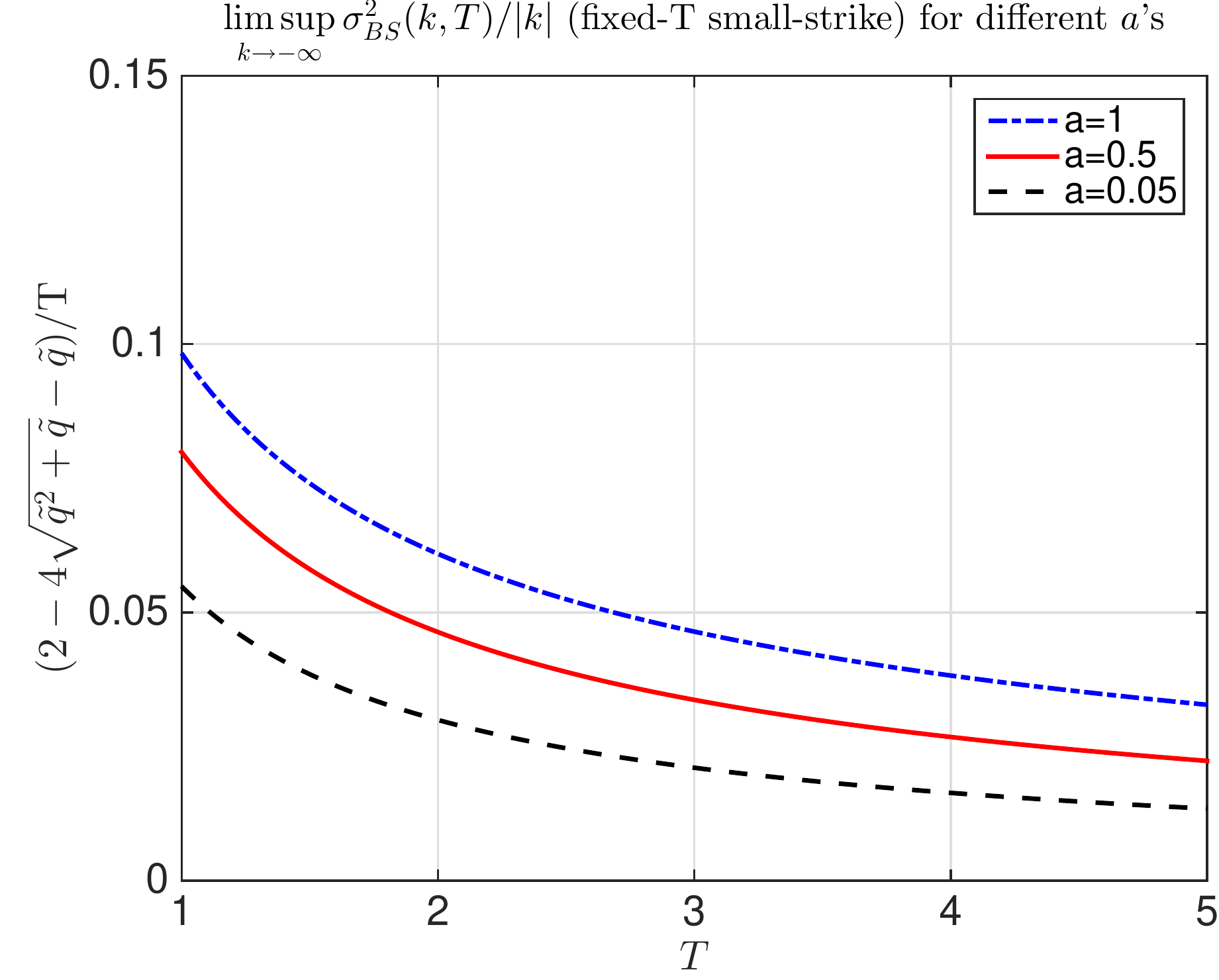}
	} 
	\caption{Left: $\mathop{\limsup}\limits_{k\rightarrow+\infty}\frac{\sigma^{2}_{\text{BS}}(k,T)}{|k|}$ (fixed-maturity large-strike) for $a=0.05,~0.5$ and $1$; Right: $\mathop{\limsup}\limits_{k\rightarrow-\infty}\frac{\sigma^{2}_{\text{BS}}(k,T)}{|k|}$ (fixed-maturity small-strike) for $a=0.05,~0.5$ and $1$.}
	\label{fig:fixed_T_small_large_strike_a}
\end{figure}

\begin{figure}[H]
	\centering 
	\subfigure{
	\includegraphics[scale=0.35]{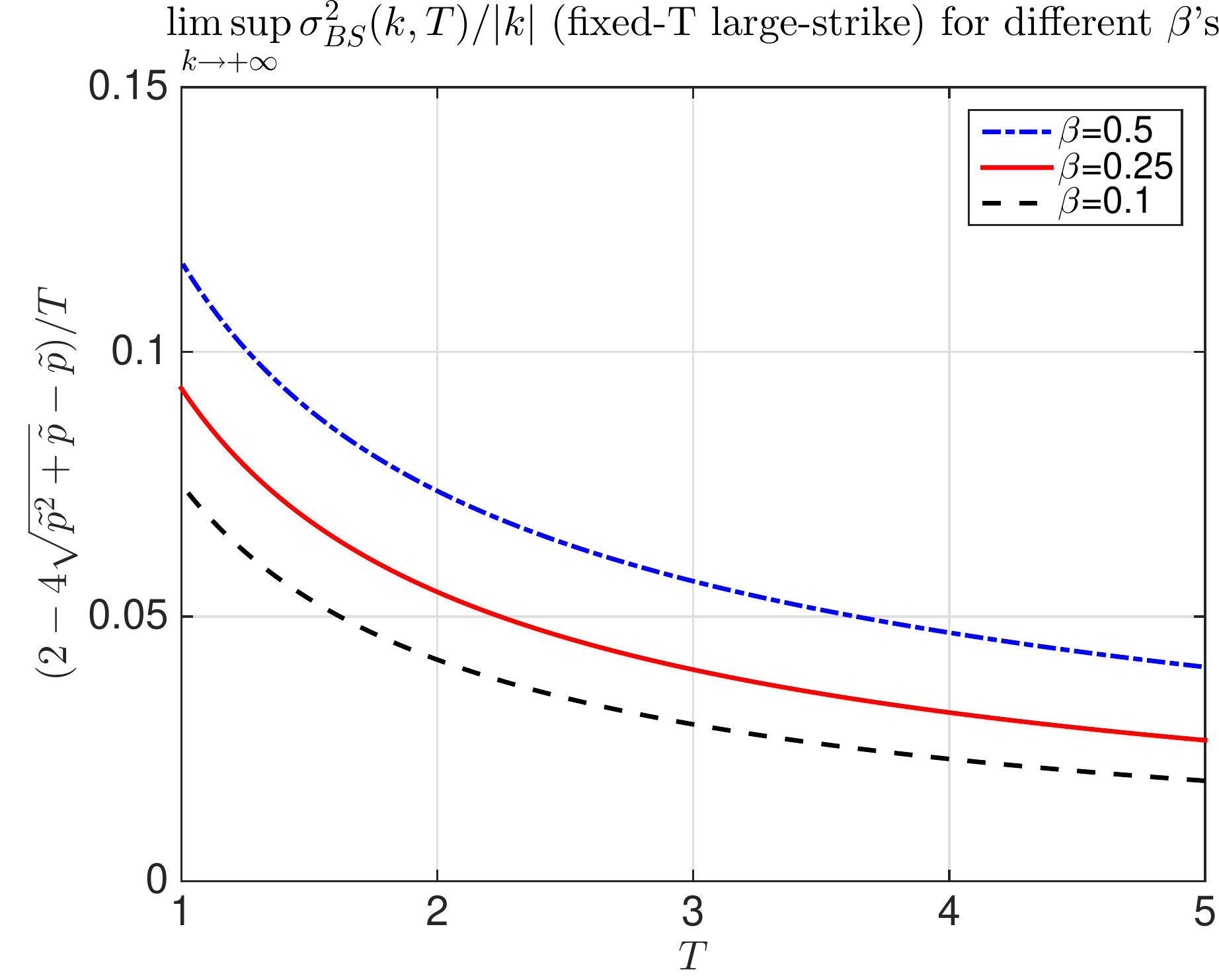}
	} 
	\subfigure{ 
	\includegraphics[scale=0.35]{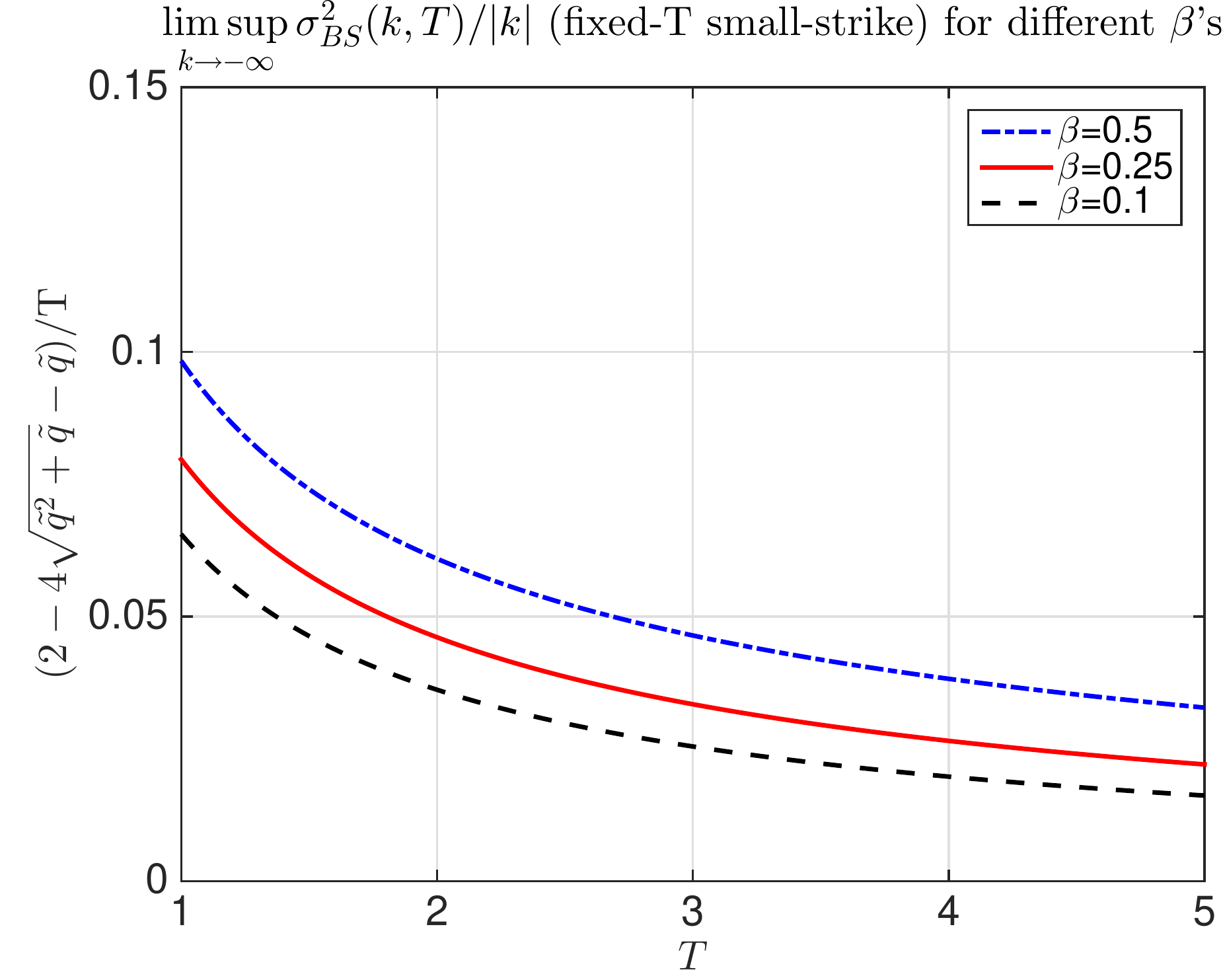}
	} 
	\caption{Left: $\mathop{\limsup}\limits_{k\rightarrow+\infty}\frac{\sigma^{2}_{\text{BS}}(k,T)}{|k|}$ (fixed-maturity large-strike) for $\beta=0.1,~0.25$ and $0.5$; Right: $\mathop{\limsup}\limits_{k\rightarrow-\infty}\frac{\sigma^{2}_{\text{BS}}(k,T)}{|k|}$ (fixed-maturity small-strike) for $\beta=0.1,~0.25$ and $0.5$.}
	\label{fig:fixed_T_small_large_strike_beta}
\end{figure}

\section{Concluding Remarks}
${}\quad$ In this paper, we study the asymptotic behaviors of the implied volatility of an affine jump-diffusion model.
Let $X_t=\log(S_t/S_0)$ and $S_t$ follows an affine jump-diffusion model under risk-neutral measure.
By applying the Feynman-Kac formula, we compute the moment generating function for $X_t$. An explicit form of the moment generating function can be found by solving a set of ordinary differential equations. 
A large-maturity large deviation principle for $X_t$ is obtained by using the G\"{a}rtner-Ellis Theorem.
We characterize the asymptotic behaviors of implied volatility for $X_t$ in the joint regime of large-maturity and large-strike regime.
We use Lee's moment formula to derive the asymptotics for Black-Scholes implied volatility in the fixed-maturity, large-strike and fixed-maturity, small-strike regimes. 
Numerical studies are provided to validate the theoretical work.
We observe the volatility smiles in the joint regime of large-maturity and large-strike.
As the self-exciting intensity parameter ($a$ or $\beta$) increases, which means more rare events tending to occur, the ATM volatility increases and volatility smile tends to be more convex.
Ratios of Black-Scholes implied volatility to log-moneyness in fixed-maturity large, small-strike and large-strike regimes are shown.
For a given maturity $T$, as the self-exciting parameter ($a$ or $\beta$) increases, the ratio of implied volatility to log-moneyness increases.
In these two regimes, we observe the ratio of implied volatility to log-moneyness declines as the maturity increases and this is usually detected on an implied volatility surface in practice.

\section*{Acknowledgments}
Nian Yao was supported in part by Natural Science Foundation of Guangdong Province under Grant 2019A1515012192. Author Nian Yao and Zhiqiu Li acknowledge Dr. Lingjiong Zhu for many useful discussions and comments.


\bibliographystyle{apalike}
\bibliography{ref}

\end{document}